\documentclass[runningheads]{llncs}

\usepackage{color}
\usepackage{listings}
\usepackage{url}
\usepackage{enumitem}
\usepackage{mathtools}
\usepackage{braket}
\usepackage{amsmath}
\usepackage{todonotes}
\usepackage{makecell}
\usepackage{amsfonts}
\usepackage[hidelinks]{hyperref}

\usepackage{tabularx,array,booktabs}

\definecolor{graybright}{cmyk}{0, 0, 0, 0.15}
\definecolor{mydarkgreen}{cmyk}{0.85, 0.31, 0.96, 0.2}
\definecolor{myblue}{cmyk}{0.91, 0.67, 0.53, 0.51}
\definecolor{myred}{cmyk}{0.08, 0.86, 0.75, 0.01}
\definecolor{myorange}{cmyk}{0.08, 0.49, 1, 0}
\definecolor{stringgreen}{rgb}{0.25,0.5,0.35} 
\definecolor{mygray}{rgb}{0.5, 0.5, 0.5} 

\usepackage{pifont}
\newcommand{\cmark}{\color{mydarkgreen} \ding{51}}%
\newcommand{\xmark}{\color{myred} \ding{55}}%

\usepackage{tikz}
\usetikzlibrary{automata, positioning, arrows, fit}
\tikzset{->,  
>=stealth, 
node distance=0.1cm, 
}

\tikzset{every state/.append style={thick, font=\scriptsize, fill=gray!10, rectangle, rounded corners, inner sep=3pt, minimum size=0.4cm}}

\tikzset{label_node/.style={font=\scriptsize, align=center, auto}}

\newcommand*{\Scale}[2][4]{\scalebox{#1}{$#2$}}%

\providecommand{\incomp}{
\Scale[0.7]{
  \not\mathrel{
  \smash{
  \vcenter{
    \offinterlineskip 
    \ialign{
       \hfil##\hfil\cr 
       $\sqsubseteq$\cr 
       \noalign{\kern.3ex}
       $\sqsupset$\cr 
    }
  }
  }
  \vphantom{>}
  }
 }
}

\newcommand{\incomparable}{\raisebox{0.05cm}{\incomp}}

\newcommand{\CodeCurly}[1]{\textcolor{mydarkgreen}{#1}}
\newcommand{\CodeRound}[1]{\textcolor{myred}{#1}}

\usepackage{adjustbox}

\usepackage{letltxmacro}
\newcommand*{\SavedLstInline}{}
\LetLtxMacro\SavedLstInline\lstinline
\DeclareRobustCommand*{\lstinline}{%
  \ifmmode
    \let\SavedBGroup\bgroup
    \def\bgroup{%
      \let\bgroup\SavedBGroup
      \hbox\bgroup
    }%
  \fi
  \SavedLstInline
}

\lstdefinestyle{BCSL}{
  literate={(}{{\CodeRound{(}}}1
           {=>}{{$\Rightarrow$}}1
           {+}{{\hspace{0.05cm}$+$\hspace{0.05cm}}}1
           {::}{{\hspace{0.01cm}$::$\hspace{0.01cm}}}1
           {*}{{$\times$}}1
         {)}{{\CodeRound{)}}}1
         {\{}{{\CodeCurly{\{}}}1
           {\}}{{\CodeCurly{\}}}}1, 
  commentstyle=\color{mygray}\ttfamily
}

\lstdefinestyle{EBNF}{
  stringstyle=\color{stringgreen},
  showstringspaces=false
}

\usepackage{geometry}
\begin{document}

\title{Regulated Multiset Rewriting Systems} 

\titlerunning{Regulated Multiset Rewriting Systems} 

\author{Matej Troj\'{a}k, Samuel Pastva, David \v{S}afr\'{a}nek, and Lubo\v{s} Brim}

\institute{Faculty of Informatics, Masaryk University, Brno, Czech Republic}

\authorrunning{Troj\'ak et al.} 

\maketitle

\begin{abstract}
Multiset rewriting systems provide a formalism particularly suitable for the description of biological systems. We present an extension of this formalism with additional controls on the derivations as a tool for reducing possible non-deterministic behaviour by providing additional knowledge about the system. We introduce several regulation mechanisms and compare their generative power.
\end{abstract}

\section{Introduction}

Multiset rewriting systems (MRSs)~\cite{meseguer1992conditional} provide a modelling formalism suitable in particular for describing biological systems~\cite{barbuti2008intermediate,bistarelli2003representing}. MRS can be considered as a low-level formalism, describing how multisets of elements can be modified by applying rewriting \emph{rules} representing biological processes. Such a basic mechanism can serve as a base to build higher-level languages (e.g. rule-based modelling~\cite{boutillier2018kappa,danos2013constraining,harris2016bionetgen}, membrane computing~\cite{puaun2007membrane}, Petri nets~\cite{cervesato1994petri}), providing more convenient approaches to describe interesting biological phenomena. The mechanistic nature of MRS supports extendible and decomposable system construction by separate rule description, which becomes extremely useful when newly discovered biological knowledge needs to be introduced into the system.

Due to the non-deterministic behaviour of MRS, there are types of knowledge that cannot be simply covered by adding new rules or modifying existing ones. For example, consider a situation where a rule is always followed by another one (e.g. complex formation and its phosphorylation in circadian clock mechanism~\cite{dved2016formal}), two rules cannot take place after each other (e.g. cell division after entering apoptosis phase~\cite{swat2004bifurcation}), a rule always has a priority over another one (e.g. cell differentiation~\cite{harel2008concurrency}), or a rule cannot be used in a particular context (e.g. explicit modelling of inhibition in a signalling pathway).

A possible approach to capture the above-described types of behaviour is to employ some rewriting control in the form of a \emph{regulated} application of rewriting rules. Regulation of rewriting thus could provide an additional mechanism for controlling the rule application by restricting the conditions when a rule can be enabled. Regulated rewriting is well-established in the area of string grammars~\cite{dassow2004grammars} and has been extended to several other modelling formalisms such as Petri nets~\cite{ichikawa1988analysis} or membrane computing~\cite{freund2004regulated}. In~\cite{delzanno2002overview} authors introduced an approach to regulate MRS using first-order logical formulae over variables, allowing to express constraints to rules application. While this approach significantly increases the expressive power of MRS, it is not particularly suitable for reasoning about the type of behaviour based on relationships among rewriting rules as described in the examples above.

In this paper, we introduce separate classes of regulation mechanisms of \emph{multiset rewriting systems}. They are based on our experience with modelling in systems biology~\cite{vsafranek2011photosynthesis,trojak2016cyanobacterium,trojak2020plosone} and correspond to the form of specification experts outside of computer science expect. We define several regulation mechanisms, namely: \emph{regular} rewriting restricted by regular language over rules, \emph{ordered} rewriting restricted by strict partial order on rules, \emph{programmed} rewriting restricted by successors for each rule, \emph{conditional} rewriting restricted by prohibited context for each rule, and \emph{concurrent-free} rewriting restricted by resolving concurrent behaviour of multiple rules. We investigate and compare the generative power of the defined types of regulations and show the effect of regulations on the expressive power (both summarised in Figure~\ref{summary}).

\section{Multiset rewriting systems}

In this section, we recall some definitions and known results about multisets and rewriting systems over them. Intuitively, a multiset is a set of elements with allowed repetitions. A~multiset rewriting \emph{rule} describes how a particular multiset is transformed into another one. A multiset rewriting system consists of a set of multiset rewriting rules, defining how the system can evolve, and an initial multiset, representing the starting point for the rewriting.

Let $\mathcal{S}$ be a finite set of \emph{elements}. A \emph{multiset} over $\mathcal{S}$ is a total function $ \texttt{M} : \mathcal{S} \rightarrow \mathbb{N}$ (where $\mathbb{N}$ is the set of natural numbers including 0). For each  $\mathit{a} \in \mathcal{S}$ the \emph{multiplicity} (the number of occurrences) of $\mathit{a}$ is the number $\mathtt{M}(\mathit{a})$. Operations and relations over multisets are defined in a standard way, taking into account the repetition of elements (for details we refer to~\cite{cervesato1994petri}).

A \emph{multiset rewriting rule} over $\mathcal{S}$ is a pair $\mu = (^{\bullet}\mu, \mu^{\bullet})$ of multisets over $\mathcal{S}$, written as $\mu: \hspace{0.01cm} ^{\bullet}\mu \to \mu^{\bullet}$ for convenience. The rule rewrites elements specified in the left-hand side $^{\bullet}\mu$ to elements specified in the right-hand side $\mu^{\bullet}$.

A \emph{multiset rewriting system} (MRS) over $\mathcal{S}$ is a pair $\mathcal{M} = (\mathcal{X}, \mathtt{M_0})$, where $\mathcal{X}$ is a finite set of multiset rewrite rules and $\mathtt{M_0}$ is the initial multiset (\emph{state}), both over $\mathcal{S}$.

We denote by $\mathbb{MRS}$ the class of (non-regulated) multiset rewriting systems. An example of an MRS is given in Figure~\ref{mrs_example}.

\begin{figure}[!h]
\begin{center}
 $\mathcal{M} = 
		\left\{ 
			\begin{array}{l}
	  		\mathtt{M_0 = \{A, A\}} \\
	  		\mathcal{X} = \left\{ 
	  		 	\begin{array}{l}
	  		 	\mu_\mathtt{1}: \mathtt{ \{A,A\} } \to \mathtt{ \{B,A\} }, \mu_\mathtt{2}: \mathtt{ \{A\} } \to \mathtt{ \{X\} } \\
	  		 	\mu_\mathtt{3}: \mathtt{ \{B\} } \to \mathtt{ \{C\} }, \mu_\mathtt{4}: \mathtt{ \{C\} } \to \mathtt{ \{B\} }
	  		 	\end{array} 
	  		  \right\}
	  		\end{array}
  		\right\}$
\end{center}
\caption{Example of an MRS over elements $\mathcal{S} = \mathtt{\{ A,B,C,X \}}$.}\label{mrs_example}
\end{figure}

Let $\mathtt{M}$ be a multiset and $\mathcal{M} = (\mathcal{X}, \mathtt{M_0})$ an MRS, both over $\mathcal{S}$.
A rule $\mu \in \mathcal{X}$ is \emph{enabled} at $\mathtt{M}$ if $^{\bullet}\mu \subseteq \mathtt{M}$. 
The \emph{application} of an enabled rule $\mu \in \mathcal{X}$ to $\mathtt{M}$, written $\mathtt{M} \rightarrow_{\mu} \mathtt{M}'$, creates a multiset $\mathtt{M}' = (\mathtt{M} \setminus \hspace{0.01cm} ^{\bullet}\mu) \cup \mu^{\bullet}$.
A \emph{run} $\pi$ of $\mathcal{M}$ is an infinite sequence of multisets $\pi = \mathtt{M_0} \mathtt{M_1} \mathtt{M_2} \ldots$ such that for any \emph{step} $\mathtt{i > 0}$ holds that $\mathtt{M_{i-1}} \rightarrow_{\mu} \mathtt{M_i}$ for some $\mu \in \mathcal{X}$. We denote by $\pi[\mathtt{i}]$ the multiset created in step $\mathtt{i}$. A \emph{run label} $\overrightarrow{\pi}$ of a run $\pi$ is an infinite sequence of rules $\overrightarrow{\pi} = \mu_\mathtt{1} \mu_\mathtt{2} \mu_\mathtt{3} \ldots$ such that for any \emph{step} $\mathtt{i > 0}$ holds that $\pi[\mathtt{i-1}] \rightarrow_{\mu_\mathtt{i}} \pi[\mathtt{i}]$. We denote by $\overrightarrow{\pi}[\mathtt{i}]$ the rule applied in step \texttt{i}.

The \emph{semantics} of the system $\mathcal{M}$ is a (potentially infinite) set $\mathfrak{L}(\mathcal{M})$ of all possible runs such that $\forall \pi \in \mathfrak{L}(\mathcal{M}).~ \pi[\mathtt{0}] = \mathtt{M_0}$ (runs start in the initial multiset). In order to ensure the infiniteness of runs, we implicitly assume the presence of a special \emph{empty} rule $\varepsilon = (\emptyset, \emptyset)$. We require that this rule can be applied \emph{only} when no other rule of the system is enabled. It also ensures that the set of rules $\mathcal{X}$ is always non-empty. In Figure~\ref{runs_example} we give an example of runs.

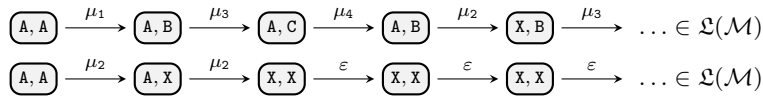
\begin{figure}[!h]
\begin{center}
\begin{tikzpicture}

\node[state] (s0) {$\mathtt{A,A}$};
\node[state, right=1cm of s0] (s1) {$\mathtt{A,B}$};
\node[state, right=1cm of s1] (s2) {$\mathtt{A,C}$};
\node[state, right=1cm of s2] (s3) {$\mathtt{A,B}$};
\node[state, right=1cm of s3] (s4) {$\mathtt{X,B}$};
\node[right=1cm of s4] (D) {$\ldots \in \mathfrak{L}(\mathcal{M})$};

\draw[shorten >=1mm,shorten <=1mm]
	    (s0) edge[above] node[label_node]{$\mu_\mathtt{1}$} (s1)
		(s1) edge[above] node[label_node]{$\mu_\mathtt{3}$} (s2)
		(s2) edge[above] node[label_node]{$\mu_\mathtt{4}$} (s3)
		(s3) edge[above] node[label_node]{$\mu_\mathtt{2}$} (s4)
		(s4) edge[above] node[label_node]{$\mu_\mathtt{3}$} (D);
\end{tikzpicture}

\begin{tikzpicture}

\node[state] (s0) {$\mathtt{A,A}$};
\node[state, right=1cm of s0] (s1) {$\mathtt{A,X}$};
\node[state, right=1cm of s1] (s2) {$\mathtt{X,X}$};
\node[state, right=1cm of s2] (s3) {$\mathtt{X,X}$};
\node[state, right=1cm of s3] (s4) {$\mathtt{X,X}$};
\node[right=1cm of s4] (D) {$\ldots \in \mathfrak{L}(\mathcal{M})$};

\draw[shorten >=1mm,shorten <=1mm]
	    (s0) edge[above] node[label_node]{$\mu_\mathtt{2}$} (s1)
		(s1) edge[above] node[label_node]{$\mu_\mathtt{2}$} (s2)
		(s2) edge[above] node[label_node]{$\varepsilon$} (s3)
		(s3) edge[above] node[label_node]{$\varepsilon$} (s4)
		(s4) edge[above] node[label_node]{$\varepsilon$} (D);
\end{tikzpicture}
\end{center}
\caption{Example of runs of the MRS from Figure~\ref{mrs_example}, including their run labels.}\label{runs_example}
\end{figure}

\section{Regulated rewriting}

The semantics of MRS is a set of all possible runs induced by its non-deterministic nature. In this section, we introduce \emph{regulated} MRS, which restricts the applicability of rules, leading to a reduced set of possible runs (reduced semantics).

A \emph{regulated} multiset rewriting system (rMRS) is a triple $\overline{\mathcal{M}} =(\mathcal{X}, \mathtt{M_0}, \zeta)$ such that $\mathcal{M} = (\mathcal{X}, \mathtt{M_0})$ is an MRS and $\zeta$ is a \emph{regulation}, which is used to reduce the set of possible runs $\mathfrak{L}(\mathcal{M})$ to $\mathfrak{L}(\overline{\mathcal{M}})$. The exact definition of $\zeta$ depends on the particular reduction mechanisms and is specified for individual types of suggested regulation below.

\subsection{Regular rewriting}

Informally, the idea is to define an $\omega$-regular language $\zeta$ over rules, that is, a regular language over infinite words. Then, only runs with the run label from this language are allowed. The set of words $\zeta$, due to its infiniteness, is almost always (though not necessarily) described via some more convenient mechanism (such as $\omega$-regular expression).

In a \emph{regular} multiset rewriting system, the regulation $\zeta$ is defined as an $\omega$-regular language over set of rules $\mathcal{X}$. The set of possible runs $\mathfrak{L}(\mathcal{M})$ is reduced to a set of \emph{regular} runs $\mathfrak{L}(\overline{\mathcal{M}}) = \{ ~\pi \in \mathfrak{L}(\mathcal{M}) ~|~ \overrightarrow{\pi} \in \zeta~ \}$. Moreover, we require that $\forall \overrightarrow{\pi} \in \zeta ~ \exists \pi \in \mathfrak{L}(\overline{\mathcal{M}})$ (there are no extra words in $\zeta$).

We denote by $\mathbb{RR}$ the class of regular multiset rewriting systems. An example of an $\mathbb{RR}$ system is available in Figure~\ref{rr_example} with an example of valid and invalid runs in Figure~\ref{rr_runs_example}.

\begin{figure}[!h]
\begin{center}
$\overline{\mathcal{M}} = 
		\left\{ 
			\begin{array}{l}
	  		\mathtt{M_0 = \emptyset}, \hspace{0.2cm} \zeta = (\mu_\mathtt{1}.\mu_\mathtt{2})^*.{\mu_\mathtt{3}}^*.\varepsilon^\omega,\\
	  		\mathcal{X} = \left\{ 
	  		 	\begin{array}{l}
	  		 	\mu_\mathtt{1}: \emptyset \to \mathtt{ \{A\} }, \mu_\mathtt{2}: \mathtt{ \{A\} } \to \mathtt{ \{B\} }, \mu_\mathtt{3}: \mathtt{ \{B\} } \to \emptyset
	  		 	\end{array} 
	  		  \right\} \\
	  		\end{array}
  		\right\}$
\end{center}
\caption{Example of a regular multiset rewriting system over a set of elements $\mathcal{S} = \mathtt{\{ A,B \}}$. The regulation $\zeta$ allows a particular finite sequence of rules followed by infinite application of rule $\varepsilon$.}\label{rr_example}
\end{figure}

\begin{figure}[!h]
\begin{center}
\begin{tikzpicture}
\node[state] (s0) {};
\node[state, right=1cm of s0] (s1) {$\mathtt{A}$};
\node[state, right=1cm of s1] (s2) {$\mathtt{B}$};
\node[state, right=1cm of s2] (s3) {};
\node[state, right=1cm of s3] (s4) {};
\node[right=1cm of s4] (D) {$\ldots \pi_\mathtt{1}$ \cmark};

\draw[shorten >=1mm,shorten <=1mm]
	    (s0) edge[above] node[label_node]{$\mu_\mathtt{1}$} (s1)
		(s1) edge[above] node[label_node]{$\mu_\mathtt{2}$} (s2)
		(s2) edge[above] node[label_node]{$\mu_\mathtt{3}$} (s3)
		(s3) edge[above] node[label_node]{$\varepsilon$} (s4)
		(s4) edge[above] node[label_node]{$\varepsilon$} (D);
\end{tikzpicture}

\begin{tikzpicture}

\node[state] (s0) {};
\node[state, right=1cm of s0] (s1) {$\mathtt{A}$};
\node[state, right=1cm of s1] (s2) {$\mathtt{B}$};
\node[state, right=1cm of s2] (s3) {};
\node[state, right=1cm of s3] (s4) {$\mathtt{A}$};
\node[right=1cm of s4] (D) {$\ldots \pi_\mathtt{2}$ \xmark};

\draw[shorten >=1mm,shorten <=1mm]
	    (s0) edge[above] node[label_node]{$\mu_\mathtt{1}$} (s1)
		(s1) edge[above] node[label_node]{$\mu_\mathtt{2}$} (s2)
		(s2) edge[above] node[label_node]{$\mu_\mathtt{3}$} (s3)
		(s3) edge[above] node[label_node]{$\color{myred} \mu_\mathtt{1}$} (s4)
		(s4) edge[above] node[label_node]{$\mu_\mathtt{2}$} (D);
\end{tikzpicture}
\end{center}
\caption{Example of valid and invalid runs of the $\mathbb{RR}$ system from Figure~\ref{rr_example}. Run $\pi_\mathtt{1} \in \mathfrak{L}(\overline{\mathcal{M}})$ because $\protect\overrightarrow{\pi_\mathtt{1}}~=~\mu_\mathtt{1}.\mu_\mathtt{2}.\mu_\mathtt{3}.\varepsilon^\omega \in \zeta$ and run $\pi_\mathtt{2} \not\in \mathfrak{L}(\overline{\mathcal{M}})$ because $\protect\overrightarrow{\pi_\mathtt{2}} = \mu_\mathtt{1}.\mu_\mathtt{2}.\mu_\mathtt{3}.\mu_\mathtt{1}\ldots \not\in \zeta$.}\label{rr_runs_example}
\end{figure}

\subsection{Ordered rewriting}

The ordered rewriting regulation is based on a strict partial order on rules. Thus, only runs with the run label, which does not break the ordering in the immediate successors, are allowed.

In an \emph{ordered} multiset rewriting system, the regulation $\zeta$ is a strict partial order over a set of rules. For any two rules $\mu, \mu' \in \mathcal{X}$ we write $\mu < \mu'$ iff $(\mu, \mu') \in \zeta$. The set of possible runs $\mathfrak{L}(\mathcal{M})$ is reduced to a set of \emph{ordered} runs $\mathfrak{L}(\overline{\mathcal{M}}) = \{ ~\pi \in \mathfrak{L}(\mathcal{M}) ~|~ \forall \mathtt{i > 0} .~ \overrightarrow{\pi}[\mathtt{i+1}] \not< \overrightarrow{\pi}[\mathtt{i}] \}$.

We denote by $\mathbb{OR}$ the class of ordered multiset rewriting systems. An example of an $\mathbb{OR}$ system is available in Figure~\ref{or_example} with an example of valid and invalid runs in Figure~\ref{or_runs_example}.

\begin{figure}[!h]
\begin{center}
$\overline{\mathcal{M}} = 
		\left\{ 
			\begin{array}{l}
	  		\mathtt{M_0 = \emptyset}, \hspace{0.2cm} \zeta = \left\{ 
          \begin{array}{l}
          (\mu_\mathtt{1}, \mu_\mathtt{2})
          \end{array} 
          \right\}, \\
	  		\mathcal{X} = \left\{ 
	  		 	\begin{array}{l}
	  		 	\mu_\mathtt{1}: \emptyset \to \mathtt{ \{A\} }, \mu_\mathtt{2}: \mathtt{ \{A\} } \to \emptyset
	  		 	\end{array} 
	  		  \right\} \\
	  		\end{array}
  		\right\}$
\end{center}
\caption{Example of an ordered multiset rewriting system over a set of elements $\mathcal{S} = \mathtt{\{ A \}}$. The regulation $\zeta$ defines order $\mu_\mathtt{1} < \mu_\mathtt{2}$ on rules $\mu_\mathtt{1}, \mu_\mathtt{2}$, which does not allow $\mu_\mathtt{1}$ to be immediately used after $\mu_\mathtt{2}$.}\label{or_example}
\end{figure}

\begin{figure}[!h]
\begin{center}
\begin{tikzpicture}

\node[state] (s0) {};
\node[state, right=1cm of s0] (s1) {$\mathtt{A}$};
\node[state, right=1cm of s1] (s2) {$\mathtt{A,A}$};
\node[state, right=1cm of s2] (s3) {$\mathtt{A}$};
\node[state, right=1cm of s3] (s4) {};
\node[right=1cm of s4] (D) {$\ldots \pi_\mathtt{1}$ \cmark};

\draw[shorten >=1mm,shorten <=1mm]
	    (s0) edge[above] node[label_node]{$\mu_\mathtt{1}$} (s1)
		(s1) edge[above] node[label_node]{$\mu_\mathtt{1}$} (s2)
		(s2) edge[above] node[label_node]{$\mu_\mathtt{2}$} (s3)
		(s3) edge[above] node[label_node]{$\mu_\mathtt{2}$} (s4)
		(s4) edge[above] node[label_node]{$\varepsilon$} (D);
\end{tikzpicture}

\begin{tikzpicture}

\node[state] (s0) {};
\node[state, right=1cm of s0] (s1) {$\mathtt{A}$};
\node[state, right=1cm of s1] (s2) {};
\node[state, right=1cm of s2] (s3) {$\mathtt{A}$};
\node[state, right=1cm of s3] (s4) {};
\node[right=1cm of s4] (D) {$\ldots \pi_\mathtt{2}$ \xmark};

\draw[shorten >=1mm,shorten <=1mm]
	    (s0) edge[above] node[label_node]{$\mu_\mathtt{1}$} (s1)
		(s1) edge[above] node[label_node]{$\color{myred} \mu_\mathtt{2}$} (s2)
		(s2) edge[above] node[label_node]{$\color{myred} \mu_\mathtt{1}$} (s3)
		(s3) edge[above] node[label_node]{$\mu_\mathtt{2}$} (s4)
		(s4) edge[above] node[label_node]{$\mu_\mathtt{1}$} (D);
\end{tikzpicture}
\end{center}
\caption{Example of valid and invalid runs of the $\mathbb{OR}$ system from Figure~\ref{or_example}. Run $\pi_\mathtt{1} \in \mathfrak{L}(\overline{\mathcal{M}})$ because the rule $\mu_\mathtt{1}$ is never used immediately after rule $\mu_\mathtt{2}$. Run $\pi_\mathtt{2} \not\in \mathfrak{L}(\overline{\mathcal{M}})$ because the rule $\mu_\mathtt{1}$ was used immediately after rule $\mu_\mathtt{2}:~\protect\overrightarrow{\pi_\mathtt{2}}[\mathtt{3}] < \protect\overrightarrow{\pi_\mathtt{2}}[\mathtt{2}]$.}\label{or_runs_example}
\end{figure}
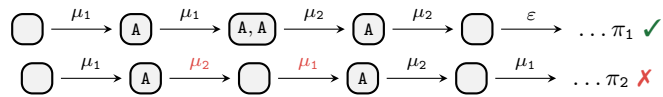

\subsection{Programmed rewriting}

Informally, for every rule, there exists a set of successor rules. Then, only runs with the run label where each rule is followed by its successor are allowed.

In a \emph{programmed} multiset rewriting system, the regulation $\zeta: \mathcal{X} \to \mathtt{2}^{\mathcal{X}}$ is a \emph{successor} function, assigning to each rule a set of successor rules. The set of possible runs $\mathfrak{L}(\mathcal{M})$ is reduced to a set of \emph{programmed} runs $\mathfrak{L}(\overline{\mathcal{M}}) = \{ ~\pi \in \mathfrak{L}(\mathcal{M}) ~|~ \forall \mathtt{i > 0} .~ \overrightarrow{\pi}[\mathtt{i+1}] \in \zeta(\overrightarrow{\pi}[\mathtt{i}]) ~\}$.

We denote by $\mathbb{PR}$ the class of programmed multiset rewriting systems. An example of a $\mathbb{PR}$ system is available in Figure~\ref{pr_example} with an example of valid and invalid runs in Figure~\ref{pr_runs_example}.

\begin{figure}[!h]
\begin{center}
$\overline{\mathcal{M}} = 
		\left\{ 
			\begin{array}{l}
	  		\mathtt{M_0 = \emptyset}, \hspace{0.2cm} \zeta = \left\{ 
          \begin{array}{l}
          \mu_\mathtt{1} \to \{\mu_\mathtt{2}\}, \mu_\mathtt{2} \to \{\mu_\mathtt{1}\}
          \end{array} 
          \right\}, \\
	  		\mathcal{X} = \left\{ 
	  		 	\begin{array}{l}
	  		 	\mu_\mathtt{1}: \emptyset \to \mathtt{ \{A\} }, \mu_\mathtt{2}: \emptyset \to \mathtt{ \{B\} }
	  		 	\end{array} 
	  		  \right\} \\
	  		\end{array}
  		\right\}$
\end{center}
\caption{Example of a programmed multiset rewriting system over a set of elements $\mathcal{S} = \mathtt{\{ A,B \}}$. The regulation $\zeta$ allows only alternate application of rules $\mu_\mathtt{1}$ and $\mu_\mathtt{2}$.}\label{pr_example}
\end{figure}

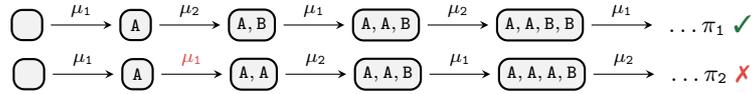
\begin{figure}[!h]
\begin{center}
\begin{tikzpicture}

\node[state] (s0) {};
\node[state, right=1cm of s0] (s1) {$\mathtt{A}$};
\node[state, right=1cm of s1] (s2) {$\mathtt{A,B}$};
\node[state, right=1cm of s2] (s3) {$\mathtt{A,A,B}$};
\node[state, right=1cm of s3] (s4) {$\mathtt{A,A,B,B}$};
\node[right=1cm of s4] (D) {$\ldots \pi_\mathtt{1}$ \cmark};

\draw[shorten >=1mm,shorten <=1mm]
	    (s0) edge[above] node[label_node]{$\mu_\mathtt{1}$} (s1)
		(s1) edge[above] node[label_node]{$\mu_\mathtt{2}$} (s2)
		(s2) edge[above] node[label_node]{$\mu_\mathtt{1}$} (s3)
		(s3) edge[above] node[label_node]{$\mu_\mathtt{2}$} (s4)
		(s4) edge[above] node[label_node]{$\mu_\mathtt{1}$} (D);
\end{tikzpicture}

\begin{tikzpicture}

\node[state] (s0) {};
\node[state, right=1cm of s0] (s1) {$\mathtt{A}$};
\node[state, right=1cm of s1] (s2) {$\mathtt{A,A}$};
\node[state, right=1cm of s2] (s3) {$\mathtt{A,A,B}$};
\node[state, right=1cm of s3] (s4) {$\mathtt{A,A,A,B}$};
\node[right=1cm of s4] (D) {$\ldots \pi_\mathtt{2}$ \xmark};

\draw[shorten >=1mm,shorten <=1mm]
	    (s0) edge[above] node[label_node]{$\mu_\mathtt{1}$} (s1)
		(s1) edge[above] node[label_node]{$\color{myred} \mu_\mathtt{1}$} (s2)
		(s2) edge[above] node[label_node]{$\mu_\mathtt{2}$} (s3)
		(s3) edge[above] node[label_node]{$\mu_\mathtt{1}$} (s4)
		(s4) edge[above] node[label_node]{$\mu_\mathtt{2}$} (D);
\end{tikzpicture}
\end{center}
\caption{Example of valid and invalid runs of the $\mathbb{PR}$ system from Figure~\ref{pr_example}. Run $\pi_\mathtt{1} \in \mathfrak{L}(\overline{\mathcal{M}})$ because rules $\mu_\mathtt{1}$ and $\mu_\mathtt{2}$ are alternating and run $\pi_\mathtt{2} \not\in \mathfrak{L}(\overline{\mathcal{M}})$ because the run label $\protect\overrightarrow{\pi_\mathtt{2}}$ violates the successor function, e.g. $\protect\overrightarrow{\pi_\mathtt{2}}[\mathtt{2}] \not\in \zeta(\protect\overrightarrow{\pi_\mathtt{2}}[\mathtt{1}])$ (i.e. $\mu_\mathtt{1} \not\in \zeta(\mu_\mathtt{1})$).}\label{pr_runs_example}
\end{figure}

\subsection{Conditional rewriting}

The idea is to define multiset of elements (or multiple of them) for each rule, which represents a \emph{prohibited} context where the rule cannot be used. Then, a run is valid, if in every step, there prohibited context of the used rule is not subset of the current multiset.

In a \emph{conditional} multiset rewriting system, the regulation $\zeta$ is defined as a \emph{conditional} function, assigning a set of multisets over elements $\mathcal{S}$ to each rule. Each of these multisets represent the prohibited context in which the rule cannot be applied, i.e. it is not a subset of the current state. The set of possible runs $\mathfrak{L}(\mathcal{M})$ is reduced to a set of \emph{conditional} runs $\mathfrak{L}(\overline{\mathcal{M}}) = \{ \pi \in \mathfrak{L}(\mathcal{M}) ~|~ \forall \mathtt{i > 0}.~\forall \mathcal{A} \in \zeta(\overrightarrow{\pi}[\mathtt{i}]).~ \mathcal{A} \not\subseteq \pi[\mathtt{i-1}] \}$.

We denote by $\mathbb{CR}$ the class of conditional multiset rewriting systems. An example of a $\mathbb{CR}$ system is available in Figure~\ref{cr_example} with an example of valid and invalid runs in Figure~\ref{cr_runs_example}. Note that in the case when there is only one prohibited context defined for a rule, we simplify the notation by omitting the parent set.

\begin{figure}[!h]
\begin{center}
$\overline{\mathcal{M}} = 
		\left\{ 
			\begin{array}{l}
	  		\mathtt{M_0 = \emptyset}, \hspace{0.2cm} \zeta = \left\{ 
          \begin{array}{l}
          \mu_\mathtt{1} \to \mathtt{ \{B\} }, \mu_\mathtt{2} \to \emptyset
          \end{array} 
          \right\}, \\
	  		\mathcal{X} = \left\{ 
	  		 	\begin{array}{l}
	  		 	\mu_\mathtt{1}: \emptyset \to \mathtt{ \{A\} }, \mu_\mathtt{2}: \mathtt{ \{A\} } \to \mathtt{ \{B\} }
	  		 	\end{array} 
	  		  \right\} \\
	  		\end{array}
  		\right\}$
\end{center}
\caption{Example of a conditional multiset rewriting system over a set of elements $\mathcal{S} = \mathtt{\{ A,B \}}$. The regulation $\zeta$ limits the application of rule $\mu_\mathtt{1}$ to states where an element \texttt{B} is not present.}\label{cr_example}
\end{figure}

\begin{figure}[!h]
\begin{center}
\begin{tikzpicture}

\node[state] (s0) {};
\node[state, right=1cm of s0] (s1) {$\mathtt{A}$};
\node[state, right=1cm of s1] (s2) {$\mathtt{A,A}$};
\node[state, right=1cm of s2] (s3) {$\mathtt{A,B}$};
\node[state, right=1cm of s3] (s4) {$\mathtt{B,B}$};
\node[right=1cm of s4] (D) {$\ldots \pi_\mathtt{1}$ \cmark};

\draw[shorten >=1mm,shorten <=1mm]
	    (s0) edge[above] node[label_node]{$\mu_\mathtt{1}$} (s1)
		(s1) edge[above] node[label_node]{$\mu_\mathtt{1}$} (s2)
		(s2) edge[above] node[label_node]{$\mu_\mathtt{2}$} (s3)
		(s3) edge[above] node[label_node]{$\mu_\mathtt{2}$} (s4)
		(s4) edge[above] node[label_node]{$\varepsilon$} (D);
\end{tikzpicture}

\begin{tikzpicture}

\node[state] (s0) {};
\node[state, right=1cm of s0] (s1) {$\mathtt{A}$};
\node[state, right=1cm of s1] (s2) {$\color{myred} \mathtt{B}$};
\node[state, right=1cm of s2] (s3) {$\mathtt{A,B}$};
\node[state, right=1cm of s3] (s4) {$\mathtt{B,B}$};
\node[right=1cm of s4] (D) {$\ldots \pi_\mathtt{2}$ \xmark};

\draw[shorten >=1mm,shorten <=1mm]
	    (s0) edge[above] node[label_node]{$\mu_\mathtt{1}$} (s1)
		(s1) edge[above] node[label_node]{$\mu_\mathtt{2}$} (s2)
		(s2) edge[above] node[label_node]{$\color{myred} \mu_\mathtt{1}$} (s3)
		(s3) edge[above] node[label_node]{$\mu_\mathtt{2}$} (s4)
		(s4) edge[above] node[label_node]{$\mu_\mathtt{1}$} (D);
\end{tikzpicture}
\end{center}
\caption{Example of valid and invalid runs of the $\mathbb{CR}$ system from Figure~\ref{cr_example}. Run $\pi_\mathtt{1} \in \mathfrak{L}(\overline{\mathcal{M}})$ because rule $\mu_\mathtt{1}$ is never used in the context of element \texttt{B}. Run $\pi_\mathtt{2} \not\in \mathfrak{L}(\overline{\mathcal{M}})$ because rule $\mu_\mathtt{1}$ was applied in prohibited context: $\mathtt{ \{B\} } \in \zeta(\protect\overrightarrow{\pi}[\mathtt{3}])$ with $\mathtt{ \{B\} } \subseteq \pi[\mathtt{2}]$.}\label{cr_runs_example}
\end{figure}
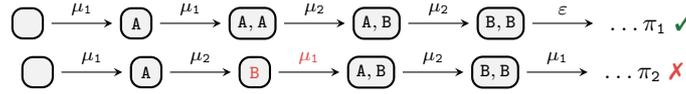

\subsection{Concurrent-free rewriting}

The idea is to detect concurrent rules and assign a priority to one of them. Then, only runs where a prioritised rule was used in place of non-prioritised ones are valid. We say two rules $\mu, \mu'$ are \emph{concurrent} iff $^\bullet\mu \hspace{0.01cm} \cap \hspace{0.01cm} ^\bullet\mu' \neq \emptyset$ (i.e., rules rewrite common elements).

In a \emph{concurrent-free} multiset rewriting system, the regulation $\zeta \subseteq \mathcal{X} \times \mathcal{X}$ is a binary relation over rules, for which holds that (1) the relation is irreflexive, (2) any two rules $\mu, \mu'$ with $(\mu, \mu') \in \zeta$ are concurrent, and (3) if $(\mu, \mu') \in \zeta$, then $(\mu', \mu) \not\in \zeta^+$, where $\zeta^+$ is the transitive closure of relation $\zeta$. The third condition ensures acyclicity, which disables ambiguous priority resolving. The set of possible runs $\mathfrak{L}(\mathcal{M})$ is reduced to a set of \emph{concurrent-free} runs 
$\mathfrak{L}(\overline{\mathcal{M}}) = \{ \pi \in \mathfrak{L}(\mathcal{M}) \mid \forall \mathtt{i} > 0~ \forall \mu \in \mathcal{X} ~\mathtt{enabled~at}~ \pi[\mathtt{i}] .~ (\overrightarrow{\pi}[\mathtt{i+1}], \mu) \not\in \zeta \}$.

We denote by $\mathbb{CFR}$ the class of concurrent-free multiset rewriting systems. An example of a $\mathbb{CFR}$ system is available in Figure~\ref{cfr_example} with an example of valid and invalid runs in Figure~\ref{cfr_runs_example}.

\begin{figure}[!h]
\begin{center}
$\overline{\mathcal{M}} = 
		\left\{ 
			\begin{array}{l}
	  		\mathtt{M_0 = \mathtt{ \{A\} } }, \hspace{0.2cm} \zeta = \left\{ 
          \begin{array}{l}
          (\mu_\mathtt{3}, \mu_\mathtt{2})
          \end{array} 
          \right\}, \\
	  		\mathcal{X} = \left\{ 
	  		 	\begin{array}{l}
          \mu_\texttt{1}: \mathtt{ \{A\} } \to \mathtt{\{ A,B \}}, \mu_\texttt{2}: \mathtt{\{A, B\}} \to \mathtt{\{ A \}}, \mu_\texttt{3}: \mathtt{\{ A \}} \to \emptyset
	  		 	\end{array} 
	  		  \right\} \\
	  		\end{array}
  		\right\}$
\end{center}
\caption{Example of a concurrent-free multiset rewriting system over a set of elements $\mathcal{S} = \mathtt{\{ A,B \}}$. The regulation $\zeta$ gives priority to rule $\mu_\mathtt{2}$ over the concurrent rule $\mu_\mathtt{3}$ in the case they are both enabled, which makes sure element \texttt{A} cannot be removed until any \texttt{B}s are present. Please note that the priority does not need to be resolved for every concurrent pair or rules (e.g. $\mu_\mathtt{1}$ and $\mu_\mathtt{2}$).}\label{cfr_example}
\end{figure}

\begin{figure}[!h]
\begin{center}
\begin{tikzpicture}

\node[state] (s0) {$\mathtt{A}$};
\node[state, right=1cm of s0] (s1) {$\mathtt{A,B}$};
\node[state, right=1cm of s1] (s2) {$\mathtt{A}$};
\node[state, right=1cm of s2] (s3) {$\mathtt{A,B}$};
\node[state, right=1cm of s3] (s4) {$\mathtt{A}$};
\node[state, right=1cm of s4] (s5) {};
\node[right=1cm of s5] (D) {$\ldots \pi_\mathtt{1}$ \cmark};

\draw[shorten >=1mm,shorten <=1mm]
	    (s0) edge[above] node[label_node]{$\mu_\mathtt{1}$} (s1)
		(s1) edge[above] node[label_node]{$\mu_\mathtt{2}$} (s2)
		(s2) edge[above] node[label_node]{$\mu_\mathtt{1}$} (s3)
		(s3) edge[above] node[label_node]{$\mu_\mathtt{2}$} (s4)
    (s4) edge[above] node[label_node]{$\mu_\mathtt{3}$} (s5)
		(s5) edge[above] node[label_node]{$\varepsilon$} (D);
\end{tikzpicture}

\begin{tikzpicture}

\node[state] (s0) {$\mathtt{A}$};
\node[state, right=1cm of s0] (s1) {$\mathtt{A,B}$};
\node[state, right=1cm of s1] (s2) {$\mathtt{A,B,B}$};
\node[state, right=1cm of s2] (s3) {$\mathtt{A,B}$};
\node[state, right=1cm of s3] (s4) {$\mathtt{B}$};
\node[state, right=1cm of s4] (s5) {$\mathtt{B}$};
\node[right=1cm of s5] (D) {$\ldots \pi_\mathtt{2}$ \xmark};

\draw[shorten >=1mm,shorten <=1mm]
	  (s0) edge[above] node[label_node]{$\mu_\mathtt{1}$} (s1)
		(s1) edge[above] node[label_node]{$\mu_\mathtt{1}$} (s2)
		(s2) edge[above] node[label_node]{$\mu_\mathtt{2}$} (s3)
		(s3) edge[above] node[label_node]{$\color{myred} \mu_\mathtt{3}$} (s4)
    (s4) edge[above] node[label_node]{$\varepsilon$} (s5)
		(s5) edge[above] node[label_node]{$\varepsilon$} (D);
\end{tikzpicture}
\end{center}
\caption{Example of valid and invalid runs of the $\mathbb{CFR}$ system from Figure~\ref{cfr_example}. Run $\pi_\mathtt{1} \in \mathfrak{L}(\overline{\mathcal{M}})$ because rule $\mu_\mathtt{2}$ was used with priority when rule $\mu_\mathtt{3}$ was also enabled. Run $\pi_\mathtt{2} \not\in \mathfrak{L}(\overline{\mathcal{M}})$ the rule $\mu_\mathtt{3}$ was used when the prioritised rule $\mu_\mathtt{2}$ was also enabled.}\label{cfr_runs_example}
\end{figure}
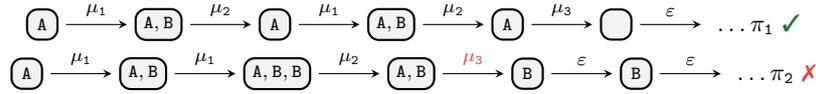

\vspace*{-0.5cm}
\begin{figure}[!h]
\begin{center}
\def\arraystretch{1.2}%
\setlength{\tabcolsep}{15pt}
\begin{tabular}{| c | c | c | c |}
\hline
\textbf{Acronym} & \textbf{Class} & \textbf{Regulation} \\ 
\hline
$\mathbb{RR}$ & regular & run labels belong to $\omega$-regular language \\
\hline
$\mathbb{OR}$ & ordered & subsequent rules are restricted by partial order \\  
\hline
$\mathbb{PR}$ & programmed & successors of rules are explicitly specified \\
\hline
$\mathbb{CR}$ & conditional & rules cannot be used in context of specified elements \\
\hline
$\mathbb{CFR}$ & concurrent-free & priority is given to one of the concurrent rules \\
\hline
\end{tabular}
\end{center}
\caption{Summary of defined regulated multiset rewriting systems.}\label{rmrs_summary}
\end{figure}

\vspace*{-0.5cm}
\begin{figure}[!h]
\begin{center}
\begin{tikzpicture}[-]
  \node (pr) at (-3,1) {$\mathbb{PR}$};
  \node (cr) at (3,1) {$\mathbb{CR}$};
  \node (or) at (-3,0) {$\mathbb{OR}$};
  \node (rr) at (0,0) {$\mathbb{RR}$};
  \node (cfr) at (3,0) {$\mathbb{CFR}$};
  \node (mrs) at (0,-1) {$\mathbb{MRS}$};
  \node (computable) at (6,0.5) {{\color{blue} computable}};
  \node (w_computable) at (-4,-1) {{\color{red} weakly computable}};
  \draw (mrs) -- (or) -- (pr)
  (rr) -- (mrs) -- (cfr) -- (cr) ;
  \node[draw=blue,anchor=west,dashed,inner sep=5pt,fit=(cr.north west)(cfr.south east)] (sq_blue) {};
  \node[draw=red,anchor=east,dashed,inner sep=5pt,fit=(mrs.north west)(mrs.south east)] (sq_red) {};
  \draw[draw=blue,dashed] (computable) -- (sq_blue);
  \draw[draw=red,dashed] (w_computable) -- (sq_red);
\end{tikzpicture}
\end{center}
\caption{Summary of generative power comparison among classes of regulated systems.}\label{summary}
\end{figure}
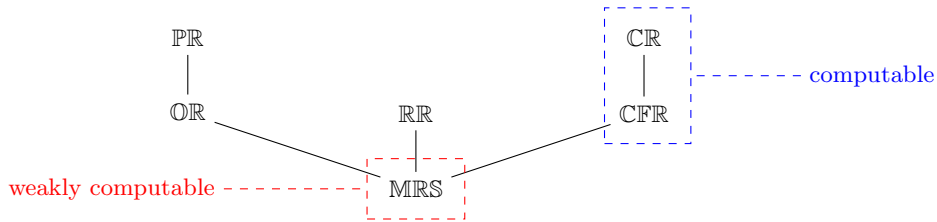

\section{Properties of regulated multiset rewriting systems}

In this section, we state some general properties of the classes of rewriting systems defined in the previous section (summarised in Figure~\ref{rmrs_summary}). First, we compare the classes on the level of generative power, i.e. runs they can generate, and then we discuss their expressive power, i.e. the functions they can compute (both summarised in Figure~\ref{summary}).

We say two runs $\pi, \pi'$ are \emph{equal} if $\forall \mathtt{i} \in \mathbb{N}.~\pi[\mathtt{i}] = \pi'[\mathtt{i}]$ (regardless of the run labels). We say two systems $\mathcal{M}, \mathcal{M}'$ are equal if their corresponding sets of runs $\mathfrak{L}(\mathcal{M}), \mathfrak{L}(\mathcal{M}')$ are equal.

Let $\mathbb{C}$ be an arbitrary class of regulated systems. We define the \emph{generative power} of a class $\mathbb{C}$ as a set $g(\mathbb{C}) = \{ \mathfrak{L}(\mathcal{M}) ~|~ \mathcal{M} \in \mathbb{C} \}$ of all possible sets of runs. We can compare classes on their level of generative power defining the following operators for any two classes $\mathbb{C}_1, \mathbb{C}_2$:

\begin{center}
\begin{minipage}{0.4\textwidth}
\begin{itemize}
	\item $\mathbb{C}_1 \sqsubseteq \mathbb{C}_2$ iff $g(\mathbb{C}_1) \subseteq g(\mathbb{C}_2)$
	\item $\mathbb{C}_1 = \mathbb{C}_2$ iff $\mathbb{C}_1 \sqsubseteq \mathbb{C}_2 \wedge \mathbb{C}_2 \sqsubseteq \mathbb{C}_1$
\end{itemize}
\end{minipage}
\begin{minipage}{0.4\textwidth}
\begin{itemize}
	\item $\mathbb{C}_1 \sqsubset \mathbb{C}_2$ iff $\mathbb{C}_1 \sqsubseteq \mathbb{C}_2 \wedge \mathbb{C}_1 \neq \mathbb{C}_2$
	\item $\mathbb{C}_1  \incomparable  \mathbb{C}_2$ iff $\mathbb{C}_1 \not\sqsubseteq \mathbb{C}_2 \wedge \mathbb{C}_2 \not\sqsubseteq \mathbb{C}_1$
\end{itemize}
\end{minipage}
\end{center}

\subsection{Generative power comparison}

All regulated classes are strictly more generative than non-regulated class.

\begin{theorem}\label{thr_mrs_vs_c}
$\forall~\mathbb{C} \in \{ \mathbb{RR}, \mathbb{OR}, \mathbb{PR}, \mathbb{CR}, \mathbb{CFR}\}.~ \mathbb{MRS} \sqsubset \mathbb{C}$ 
\end{theorem}

\begin{proof}
To show that $\mathbb{MRS} \sqsubseteq \mathbb{C}$, we show that there exists a \emph{neutral} regulation $\zeta_\mathtt{0}$ for any class $\mathbb{C}$ such that $\forall \mathcal{M} \in \mathbb{MRS} ~\exists \overline{\mathcal{M}} \in \mathbb{C}$ with $\overline{\zeta} = \zeta_\mathtt{0}$ such that $\mathfrak{L}(\mathcal{M}) = \mathfrak{L}(\overline{\mathcal{M}})$. In other words, the neutral regulation places no restriction on rules application. The neutral regulations for individual classes are defined as follows:

\vspace{-0.1cm}
\begin{center}
\begin{minipage}{0.4\textwidth}
\begin{itemize}
	\item $\mathbb{RR}: \zeta_\mathtt{0} = \mathcal{X}^\omega$
	\item $\mathbb{OR}: \zeta_\mathtt{0} = \emptyset$
	\item $\mathbb{PR}: \zeta_\mathtt{0} = \{ \mu \to \mathcal{X} ~|~ \mu \in \mathcal{X} \} $
\end{itemize}
\end{minipage}
\begin{minipage}{0.4\textwidth}
\begin{itemize}
	\item $\mathbb{CR}: \zeta_\mathtt{0} = \{ \mu \to \emptyset ~|~ \mu \in \mathcal{X} \} $
	\item $\mathbb{CFR}: \zeta_\mathtt{0} = \emptyset$
\end{itemize}
\end{minipage}
\end{center}
\vspace{-0.1cm}

To show the strictness of the relations, we find an $\overline{\mathcal{M}} \in \mathbb{C}$ such that $\forall \mathcal{M} \in \mathbb{MRS}$ holds $\mathfrak{L}(\overline{\mathcal{M}}) \neq \mathfrak{L}(\mathcal{M})$. We only show that $\mathbb{MRS} \sqsubset \mathbb{OR}$, and we skip other proofs as they are based on the same idea of showing that their respective example (used above in the class definition) does not have a representative in $\mathbb{MRS}$ with an equal set of runs.

We use $\overline{\mathcal{M}} \in \mathbb{OR}$ from Figure~\ref{or_example} and the example of valid run $\pi_\mathtt{1}$ from Figure~\ref{or_runs_example}. The rule $\mu_\mathtt{1}$ applied in step one and $\mu_\mathtt{2}$ applied in step four enforce that $\mu_\mathtt{1}, \mu_\mathtt{2} \in \mathcal{X}$ for the $\mathcal{M} \in \mathbb{MRS}$ (because there is no other way these transitions could happen). However, then a run $\pi' \in \mathfrak{L}(\mathcal{M})$:

{\centering
\begin{tikzpicture}

\node[state] (s0) {};
\node[state, right=1cm of s0] (s1) {$\mathtt{A}$};
\node[state, right=1cm of s1] (s2) {};
\node[state, right=1cm of s2] (s3) {$\mathtt{A}$};
\node[state, right=1cm of s3] (s4) {$\mathtt{A,A}$};
\node[right=1cm of s4] (D) {$\ldots$};

\draw[shorten >=1mm,shorten <=1mm]
	    (s0) edge[above] node[label_node]{$\mu_\mathtt{1}$} (s1)
		(s1) edge[above] node[label_node]{$\mu_\mathtt{2}$} (s2)
		(s2) edge[above] node[label_node]{$\mu_\mathtt{1}$} (s3)
		(s3) edge[above] node[label_node]{$\mu_\mathtt{1}$} (s4)
		(s4) edge[above] node[label_node]{$\mu_\mathtt{1}$} (D);
\end{tikzpicture}
\\}

\noindent is possible, but clearly run $\pi' \not\in \mathfrak{L}(\overline{\mathcal{M}})$. From this follows that $\mathfrak{L}(\overline{\mathcal{M}}) \neq \mathfrak{L}(\mathcal{M})$ and thus the strictness of inclusion. \qed
\end{proof}

From Figure~\ref{thr_mrs_vs_c} also follows that all regulated classes have a common intersection in $\mathbb{MRS}$. Now we show some relations among regulated classes. These are mostly incomparable, as we will show below. However, there are also two pairs of classes, namely $\mathbb{OR} \sqsubset \mathbb{PR}$ and $\mathbb{CFR} \sqsubset \mathbb{CR}$, which have a subclass relation.

\begin{theorem}\label{thr_or_vs_pr}
$\mathbb{OR} \sqsubset \mathbb{PR}$
\end{theorem}

\begin{proof}
To show that $\mathbb{OR} \sqsubseteq \mathbb{PR}$, we describe how for any $\mathcal{M} \in \mathbb{OR}$ an equivalent system $\mathcal{M}' \in \mathbb{PR}$ can be constructed. The partial order $\zeta$ defined on the rules does not allow a rule lower in the order to be used immediately after a rule higher in the order. In other words, rules higher in the order or those incomparable \emph{can} be applied. Using this fact we construct successor function $\zeta'(\mu) = \{~ \bar{\mu} \in \mathcal{X} ~|~ (\bar{\mu}, \mu) \not\in \zeta ~\}$ for all $\mu \in \mathcal{X}$ and $\mathcal{X}' = \mathcal{X}$, $\mathtt{M'_0} = \mathtt{M_0}$.

To prove that $\mathfrak{L}(\mathcal{M}) = \mathfrak{L}(\mathcal{M}')$, we must show that for any state $\mathtt{M}$ and rule $\mu \in \mathcal{X}$ it holds that $\mathcal{M}$ can apply $\mu$ to $\mathtt{M}$ $\Leftrightarrow$ $\mathcal{M'}$ can apply $\mu$ to $\mathtt{M}$. This is trivial for the initial state $\mathtt{M_0}$ since neither of the systems set any restrictions on rule application in the first step (from the definition). For any other general case 
\begin{tikzpicture}[baseline=0mm]

\node (s0) {$\ldots$};
\node[state, right=1cm of s0] (s1) {$\mathtt{M}$};
\node[right=1cm of s1] (D) {$\ldots$};

\draw[shorten >=1mm,shorten <=1mm]
	    (s0) edge[above] node[label_node]{$\mu_\mathtt{pre}$} (s1)
		(s1) edge[above] node[label_node]{$\mu$} (D);
\end{tikzpicture}
with an enabled rule $\mu$:

\begin{itemize}
	\item[$\Rightarrow$:] If $\mathcal{M}$ can apply $\mu$ to $\mathtt{M}$, that means that $(\mu, \mu_\mathtt{pre}) \not\in \zeta$. From that and definition of $\zeta'$ follows that $\mu \in \zeta'(\mu_\mathtt{pre})$. Therefore, $\mathcal{M}'$ also can apply $\mu$ to $\mathtt{M}$.

	\vspace{0.1cm}

	If $\mathcal{M}$ can \emph{not} apply $\mu$ to $\mathtt{M}$, that means that $(\mu, \mu_\mathtt{pre}) \in \zeta$. From that and definition of $\zeta'$ follows that $\mu \not\in \zeta'(\mu_\mathtt{pre})$. Therefore, $\mathcal{M}'$ also can \emph{not} apply $\mu$ to $\mathtt{M}$.

	\vspace{0.1cm}

	\item[$\Leftarrow$:] If $\mathcal{M}'$ can apply $\mu$ to $\mathtt{M}$, that means that $\mu \in \zeta'(\mu_\mathtt{pre})$. From the definition of $\zeta'$ follows that $(\mu, \mu_\mathtt{pre}) \not\in \zeta$. Therefore, $\mathcal{M}$ also can apply $\mu$ to $\mathtt{M}$.

	\vspace{0.1cm}

	If $\mathcal{M}'$ can \emph{not} apply $\mu$ to $\mathtt{M}$, that means that $\mu \not\in \zeta'(\mu_\mathtt{pre})$. From the definition of $\zeta'$ follows that $(\mu, \mu_\mathtt{pre}) \in \zeta$. Therefore, $\mathcal{M}$ also can \emph{not} apply $\mu$ to $\mathtt{M}$.
\end{itemize}

To show the strictness of the inclusion, we find an $\mathcal{M} \in \mathbb{PR}$ such that $\forall \mathcal{M}' \in \mathbb{OR}$ holds $\mathfrak{L}(\mathcal{M}) \neq \mathfrak{L}(\mathcal{M}')$. We use $\mathcal{M} \in \mathbb{PR}$ from Figure~\ref{pr_example}. The example of a valid run $\pi_\mathtt{1}$ from Figure~\ref{pr_runs_example} enforces that $\mu_\mathtt{1} \in \mathcal{X}'$ due to its first step. The symmetric run $\pi'_1$:

{\centering
\begin{tikzpicture}

\node[state] (s0) {};
\node[state, right=1cm of s0] (s1) {$\mathtt{B}$};
\node[state, right=1cm of s1] (s2) {$\mathtt{A,B}$};
\node[state, right=1cm of s2] (s3) {$\mathtt{A,B,B}$};
\node[state, right=1cm of s3] (s4) {$\mathtt{A,A,B,B}$};
\node[right=1cm of s4] (D) {$\ldots \not\in \mathfrak{L}(\mathcal{M})$};

\draw[shorten >=1mm,shorten <=1mm]
	    (s0) edge[above] node[label_node]{$\mu_\mathtt{2}$} (s1)
		(s1) edge[above] node[label_node]{$\mu_\mathtt{1}$} (s2)
		(s2) edge[above] node[label_node]{$\mu_\mathtt{2}$} (s3)
		(s3) edge[above] node[label_node]{$\mu_\mathtt{1}$} (s4)
		(s4) edge[above] node[label_node]{$\mu_\mathtt{2}$} (D);
\end{tikzpicture}
\\}

\noindent enforces that $\mu_\mathtt{2} \in \mathcal{X}'$ also due to its first step. Then, for the ordered system $\mathcal{M}'$, there are three options w.r.t. $\zeta'$ how these two rules can be related (regardless of other rules in $\mathcal{X}'$):

\vspace*{0.3cm}

\hspace{1.8cm}\begin{tabularx}{1\textwidth}{m{3cm} m{3cm}}
\multicolumn{1}{c}{\makecell{$\mu_\mathtt{1}, \mu_\mathtt{2}$ are \\ incomparable}} & 
\begin{tikzpicture}
\node[state] (s0) {};
\node[state, right=1cm of s0] (s1) {$\mathtt{A}$};
\node[state, right=1cm of s1] (s2) {$\mathtt{A,A}$};
\node[state, right=1cm of s2] (s3) {$\mathtt{A,A,A}$};
\node[right=1cm of s3] (D) {$\ldots \not\in \mathfrak{L}(\mathcal{M})$};
\draw[shorten >=1mm,shorten <=1mm]
	    (s0) edge[above] node[label_node]{$\mu_\mathtt{1}$} (s1)
		(s1) edge[above] node[label_node]{$\mu_\mathtt{1}$} (s2)
		(s2) edge[above] node[label_node]{$\mu_\mathtt{1}$} (s3)
		(s3) edge[above] node[label_node]{$\mu_\mathtt{1}$} (D);
\end{tikzpicture} \\
\multicolumn{1}{c}{$\mu_\mathtt{1} < \mu_\mathtt{2}$} & 
\begin{tikzpicture}
\node[state] (s0) {};
\node[state, right=1cm of s0] (s1) {$\mathtt{B}$};
\node[state, right=1cm of s1] (s2) {$\mathtt{A,B}$};
\node[state, right=1cm of s2] (s3) {$\mathtt{A,A,B}$};
\node[right=1cm of s3] (D) {$\ldots \not\in \mathfrak{L}(\mathcal{M})$};
\draw[shorten >=1mm,shorten <=1mm]
	    (s0) edge[above] node[label_node]{$\mu_\mathtt{2}$} (s1)
		(s1) edge[above] node[label_node]{$\mu_\mathtt{1}$} (s2)
		(s2) edge[above] node[label_node]{$\mu_\mathtt{1}$} (s3)
		(s3) edge[above] node[label_node]{$\mu_\mathtt{1}$} (D);
\end{tikzpicture} \\
\multicolumn{1}{c}{$\mu_\mathtt{2} < \mu_\mathtt{1}$} & 
\begin{tikzpicture}
\node[state] (s0) {};
\node[state, right=1cm of s0] (s1) {$\mathtt{A}$};
\node[state, right=1cm of s1] (s2) {$\mathtt{A,A}$};
\node[state, right=1cm of s2] (s3) {$\mathtt{A,A,B}$};
\node[right=1cm of s3] (D) {$\ldots \not\in \mathfrak{L}(\mathcal{M})$};
\draw[shorten >=1mm,shorten <=1mm]
	    (s0) edge[above] node[label_node]{$\mu_\mathtt{1}$} (s1)
		(s1) edge[above] node[label_node]{$\mu_\mathtt{1}$} (s2)
		(s2) edge[above] node[label_node]{$\mu_\mathtt{2}$} (s3)
		(s3) edge[above] node[label_node]{$\mu_\mathtt{2}$} (D);
\end{tikzpicture} \\

\end{tabularx}

\vspace*{0.3cm}

\noindent for each case, there is a run valid in $\mathfrak{L}(\mathcal{M}')$ which is invalid in $\mathfrak{L}(\mathcal{M})$, i.e. $\mathfrak{L}(\mathcal{M}) \neq \mathfrak{L}(\mathcal{M}')$. \qed  
\end{proof}

\begin{theorem}\label{thr_cfr_vs_cr}
$\mathbb{CFR} \sqsubset \mathbb{CR}$
\end{theorem}

\begin{proof}

To show that $\mathbb{CFR} \sqsubseteq \mathbb{CR}$, we describe how for any $\mathcal{M} \in \mathbb{CFR}$ an equivalent system $\mathcal{M'} \in \mathbb{CR}$ can be constructed. First, we set $\mathtt{M'_0} = \mathtt{M_0}$, $\mathcal{X}' = \mathcal{X}$, and $\zeta'(\mu) = \emptyset$ for all $\mu \in \mathcal{X}'$. Then, we investigate each pair $(\mu, \bar{\mu})$ in the relation $\zeta$ (i.e. $\mu, \bar{\mu}$ are concurrent) individually:

\begin{enumerate}
	\item if $^{\bullet}\mu \supseteq \hspace{0.01cm} ^{\bullet}\bar{\mu}$ \label{crf_to_cr_ad1}

	-- rule $\mu$ can never be used $\Rightarrow$ remove it from $\mathcal{X}'$
	\item if $^{\bullet}\mu \subset \hspace{0.01cm} ^{\bullet}\bar{\mu}$ ~or~ $^{\bullet}\mu, \hspace{0.01cm} ^{\bullet}\bar{\mu}$ are incomparable \label{crf_to_cr_ad2}

	-- rule $\mu$ can never be used when $\bar{\mu}$ is enabled $\Rightarrow$ $\zeta'(\mu) = \{^{\bullet}\bar{\mu}\} \cup \zeta'(\mu)$
\end{enumerate}

To prove that $\mathfrak{L}(\mathcal{M}) = \mathfrak{L}(\mathcal{M}')$, we must show that for any state $\mathtt{M}$ and enabled rule $\mu \in \mathcal{X}$ it holds that $\mathcal{M}$ can apply $\mu$ to $\mathtt{M}$ $\Leftrightarrow$ $\mathcal{M'}$ can apply $\mu$ to $\mathtt{M}$:

\begin{itemize}
	\item[$\Rightarrow$:] If $\mathcal{M}$ can apply $\mu$ to $\mathtt{M}$, that means that $\forall \mu' \in \mathcal{X} ~\mathtt{enabled~at}~ \mathtt{M} .~ (\mu, \mu') \not\in \zeta$. No forbidden context was introduced for rule $\mu$, therefore also $\mathcal{M'}$ can apply $\mu$ to $\mathtt{M}$.

	\vspace{0.1cm}

	If $\mathcal{M}$ can \emph{not} apply $\mu$ to $\mathtt{M}$, that means that $\exists \mu' \in \mathcal{X} ~\mathtt{enabled~at}~ \mathtt{M} .~ (\mu, \mu') \in \zeta$. In case (1), rule $\mu$ is removed from $\mathcal{X}'$ and therefore also $\mathcal{M'}$ can \emph{not} apply $\mu$ to $\mathtt{M}$. In case (2), the new forbidden context $^{\bullet}\mu'$ is introduced for rule $\mu$; since $\mu'$ is obviously enabled, it means that also $^{\bullet}\mu' \subset \mathtt{M}$ and therefore $\mathcal{M'}$ also can \emph{not} apply $\mu$ to $\mathtt{M}$.

	\vspace{0.1cm}

	\item[$\Leftarrow$:] If $\mathcal{M'}$ can apply $\mu$ to $\mathtt{M}$, that means that $\forall \mathcal{A} \in \zeta'(\mu)$ holds that $ \mathcal{A} \not\subseteq \mathtt{M}$. Since for case (1) obviously $\mu \in \mathcal{X}'$ and there was no forbidden context introduced in case (2), there is no concurrent rule $\mu' \in \mathcal{X}$ with $(\mu, \mu') \in \zeta$. Therefore also $\mathcal{M}$ can apply $\mu$ to $\mathtt{M}$.

	\vspace{0.1cm}

	If $\mathcal{M'}$ can \emph{not} apply $\mu$ to $\mathtt{M}$, that means that $\exists \mathcal{A} \in \zeta'(\mu)$ such that $ \mathcal{A} \subseteq \mathtt{M}$. Since for case (1) obviously $\mu \in \mathcal{X}'$, there was some forbidden context introduced in case (2). From that follows there exists a concurrent rule $\mu' \in \mathcal{X}$ enabled at $\mathtt{M}$ with $(\mu, \mu') \in \zeta$. Therefore $\mathcal{M}$ also can \emph{not} apply $\mu$ to $\mathtt{M}$.
\end{itemize}

To show the strictness of the inclusion, we find an $\mathcal{M} \in \mathbb{CR}$ such that $\forall \mathcal{M}' \in \mathbb{CFR}$ holds $\mathfrak{L}(\mathcal{M}) \neq \mathfrak{L}(\mathcal{M}')$. We use the following $\mathcal{M} \in \mathbb{CR}$:

\begin{center}
$\mathcal{M} = 
		\left\{ 
			\begin{array}{l}
	  		\mathtt{M_0 = \emptyset}, \hspace{0.2cm} \zeta = \left\{ 
          \begin{array}{l}
          \mu_\mathtt{1} \to \mathtt{ \{C\} }, \mu_\mathtt{2} \to \mathtt{ \{B\} }
          \end{array} 
          \right\}, \\
	  		\mathcal{X} = \left\{ 
	  		 	\begin{array}{l}
	  		 	\mu_\mathtt{1}: \emptyset \to \mathtt{ \{B\} }, \mu_\mathtt{2}: \emptyset \to \mathtt{ \{C\} }
	  		 	\end{array} 
	  		  \right\} \\
	  		\end{array}
  		\right\}$
\end{center}

\noindent with $\mathcal{S} = \mathtt{\{ B,C \}}$ where using one of the rules disables the other one. Therefore, the set of runs $\mathfrak{L}(\mathcal{M})$ contains only two possible runs $\pi_\mathtt{1}$ and $\pi_\mathtt{2}$:

\vspace*{0.3cm}
{\centering
\begin{tikzpicture}

\node[state] (s0) {};
\node[state, right=1cm of s0] (s1) {$\mathtt{B}$};
\node[state, right=1cm of s1] (s2) {$\mathtt{B,B}$};
\node[state, right=1cm of s2] (s3) {$\mathtt{B,B,B}$};
\node[right=1cm of s3] (D) {$\ldots \pi_\mathtt{1} \in \mathfrak{L}(\mathcal{M})$};

\draw[shorten >=1mm,shorten <=1mm]
	    (s0) edge[above] node[label_node]{$\mu_\mathtt{1}$} (s1)
		(s1) edge[above] node[label_node]{$\mu_\mathtt{1}$} (s2)
		(s2) edge[above] node[label_node]{$\mu_\mathtt{1}$} (s3)
		(s3) edge[above] node[label_node]{$\mu_\mathtt{1}$} (D);
\end{tikzpicture}

\begin{tikzpicture}

\node[state] (s0) {};
\node[state, right=1cm of s0] (s1) {$\mathtt{C}$};
\node[state, right=1cm of s1] (s2) {$\mathtt{C,C}$};
\node[state, right=1cm of s2] (s3) {$\mathtt{C,C,C}$};
\node[right=1cm of s3] (D) {$\ldots \pi_\mathtt{2} \in \mathfrak{L}(\mathcal{M})$};

\draw[shorten >=1mm,shorten <=1mm]
	    (s0) edge[above] node[label_node]{$\mu_\mathtt{2}$} (s1)
		(s1) edge[above] node[label_node]{$\mu_\mathtt{2}$} (s2)
		(s2) edge[above] node[label_node]{$\mu_\mathtt{2}$} (s3)
		(s3) edge[above] node[label_node]{$\mu_\mathtt{2}$} (D);
\end{tikzpicture}
\\}
\vspace*{0.3cm}

\noindent To ensure at least the first step in both runs, it has to hold that both rules $\mu_\mathtt{1}, \mu_\mathtt{2} \in \mathcal{X}'$ for $\mathcal{M}' \in \mathbb{CFR}$. Since these two rules are not concurrent, we cannot restrict their application at all using $\zeta'$. From that follows there will be some additional runs present in $\mathfrak{L}(\mathcal{M}')$ and therefore $\mathfrak{L}(\mathcal{M}) \neq \mathfrak{L}(\mathcal{M}')$. \qed 
\end{proof}

All the other classes are incomparable. We now investigate individual relationships. The proofs are given by showing counterexamples and using already proven relations.

\begin{theorem}\label{thr_or_vs_cfr}
$\mathbb{OR} \incomparable \mathbb{CFR}$
\end{theorem}

\begin{proof}
To show that $\mathbb{OR} \not\sqsubseteq \mathbb{CFR}$, we find an $\mathcal{M} \in \mathbb{OR}$ such that $\forall \mathcal{M}' \in \mathbb{CFR}$ holds that $\mathfrak{L}(\mathcal{M}) \neq \mathfrak{L}(\mathcal{M}')$. We use $\mathcal{M} \in \mathbb{OR}$ from Figure~\ref{or_example}. Based on proof of Figure~\ref{thr_mrs_vs_c}, both rules $\mu_\mathtt{1},\mu_\mathtt{2}$ belong to $\mathcal{X}'$. However, these rules are not concurrent, so they can be used in an arbitrary fashion in any $\mathcal{M}' \in \mathbb{CFR}$. For example, a run:

{\centering
\begin{tikzpicture}

\node[state] (s0) {};
\node[state, right=1cm of s0] (s1) {$\mathtt{A}$};
\node[state, right=1cm of s1] (s2) {};
\node[state, right=1cm of s2] (s3) {$\mathtt{A}$};
\node[state, right=1cm of s3] (s4) {$\mathtt{A,A}$};
\node[right=1cm of s4] (D) {$\ldots$};

\draw[shorten >=1mm,shorten <=1mm]
	    (s0) edge[above] node[label_node]{$\mu_\mathtt{1}$} (s1)
		(s1) edge[above] node[label_node]{$\mu_\mathtt{2}$} (s2)
		(s2) edge[above] node[label_node]{$\mu_\mathtt{1}$} (s3)
		(s3) edge[above] node[label_node]{$\mu_\mathtt{1}$} (s4)
		(s4) edge[above] node[label_node]{} (D);
\end{tikzpicture}
\\}

\noindent is possible in $\mathfrak{L}(\mathcal{M}')$ but not in $\mathfrak{L}(\mathcal{M})$, from that $\mathfrak{L}(\mathcal{M}) \neq \mathfrak{L}(\mathcal{M}')$.

To show that $\mathbb{CFR} \not\sqsubseteq \mathbb{OR}$, we find an $\mathcal{M} \in \mathbb{CFR}$ such that $\forall \mathcal{M}' \in \mathbb{OR}$ holds that $\mathfrak{L}(\mathcal{M}) \neq \mathfrak{L}(\mathcal{M}')$. We use the following $\mathcal{M} \in \mathbb{CFR}$:

\begin{center}
$\mathcal{M} = 
		\left\{ 
			\begin{array}{l}
	  		\mathtt{M_0 = \{A,B\}}, \hspace{0.2cm} \zeta = \left\{ 
          \begin{array}{l}
          (\mu_\mathtt{2}, \mu_\mathtt{1})
          \end{array} 
          \right\}, \\
	  		\mathcal{X} = \left\{ 
	  		 	\begin{array}{l}
	  		 	\mu_\mathtt{1}: \mathtt{ \{A,B\} } \to \mathtt{ \{A,C\} }, \mu_\mathtt{2}: \mathtt{ \{A\} } \to \mathtt{ \{A,B\} }
	  		 	\end{array} 
	  		  \right\} \\
	  		\end{array}
  		\right\}$
\end{center}

\noindent with $\mathcal{S} = \mathtt{\{ A,B,C \}}$. In this system, rules are alternating due to the defined priority on rule $\mu_\mathtt{1}$. The set of runs $\mathfrak{L}(\mathcal{M})$ contains only one possible run $\pi$:

{\centering
\begin{tikzpicture}

\node[state] (s0) {$\mathtt{A,B}$};
\node[state, right=1cm of s0] (s1) {$\mathtt{A,C}$};
\node[state, right=1cm of s1] (s2) {$\mathtt{A,B,C}$};
\node[state, right=1cm of s2] (s3) {$\mathtt{A,C,C}$};
\node[state, right=1cm of s3] (s4) {$\mathtt{A,B,C,C}$};
\node[right=1cm of s4] (D) {$\ldots$};

\draw[shorten >=1mm,shorten <=1mm]
	    (s0) edge[above] node[label_node]{$\mu_\mathtt{1}$} (s1)
		(s1) edge[above] node[label_node]{$\mu_\mathtt{2}$} (s2)
		(s2) edge[above] node[label_node]{$\mu_\mathtt{1}$} (s3)
		(s3) edge[above] node[label_node]{$\mu_\mathtt{2}$} (s4)
		(s4) edge[above] node[label_node]{$\mu_\mathtt{1}$} (D);
\end{tikzpicture}
\\}

\noindent which implies that at least two rules $\mu_\mathtt{I}, \mu_\mathtt{II} \in \mathcal{X}'$ such that $\mathtt{B} \in~ ^\bullet\mu_\mathtt{I} \wedge \mathtt{C} \in \mu^\bullet_\mathtt{I}$ (from the first step) and $\mathtt{B} \in \mu^\bullet_\mathtt{II}$ (from the second step). Then, similarly to proof of Figure~\ref{thr_or_vs_pr}, there are three options (incomparable and ordered in either direction) and for each option, there exists a run possible in $\mathfrak{L}(\mathcal{M}')$ but not possible in $\mathfrak{L}(\mathcal{M})$, i.e. $\mathfrak{L}(\mathcal{M}) \neq \mathfrak{L}(\mathcal{M}')$.

Another option would be to create a unique rule $\mu_\mathtt{i}: \pi[\mathtt{i-1}] \to \pi[\mathtt{i}]$ for each step $\mathtt{i}$ and placing this rule above each rule $\mu_\mathtt{j}$ with $\mathtt{j<i}$ in the order $\zeta'$. However, this would require that $\mathcal{X}'$ is infinite which the definition does not allow. \qed 
\end{proof}

\begin{theorem}\label{thr_cfr_vs_pr}
$\mathbb{CFR} \incomparable \mathbb{PR}$
\end{theorem}

\begin{proof}
The property $\mathbb{PR} \not\sqsubseteq \mathbb{CFR}$ follows directly from $\mathbb{OR} \incomparable \mathbb{CFR}$ (Figure~\ref{thr_or_vs_cfr}) and $\mathbb{OR} \sqsubset \mathbb{PR}$ (Figure~\ref{thr_or_vs_pr}).

To show that $\mathbb{CFR} \not\sqsubseteq \mathbb{PR}$, we find an $\mathcal{M} \in \mathbb{CFR}$ such that $\forall \mathcal{M}' \in \mathbb{PR}$ holds that $\mathfrak{L}(\mathcal{M}) \neq \mathfrak{L}(\mathcal{M}')$. We use $\mathcal{M} \in \mathbb{CFR}$ from Figure~\ref{cfr_example}.

First, we make a general observation about any $\mathcal{M}' \in \mathbb{PR}$: due to the presence of successor function, the rules of a $\mathbb{PR}$ system can be seen as a directed \emph{successor graph}, where vertices are the rules and edges are the possible successors. Such a graph of the $\mathcal{M}'$ system must contain at least one cycle, which (in summation) increases the number of elements \texttt{B} in the systems state. Let $c$ be the maximal number of \texttt{B}s produced by any such cycle in this graph.

Next, let $\pi_\mathtt{I} \in \mathfrak{L}(\mathcal{M})$ be a run which produces $2c$ of \texttt{B}s, then immediately discards them, and removes the element \texttt{A} with rule $\mu_\texttt{3}$. Clearly, $\mathcal{M}'$ must contain rule $\mu_\mathtt{3}$ (e.g. fifth step of run $\pi_\mathtt{1}$ in Figure~\ref{cfr_runs_example}) and this rule is applied at the end of this run (neglecting the empty rule $\varepsilon$). Additionally, since the number of \texttt{B}s produced is $2c$, the run must contain at least one aforementioned cycle from the successor graph.

Finally, let us consider another run $\pi_\mathtt{II}$ which extends the run $\pi_\mathtt{I}$ by adding one repetition of this cycle. The run $\pi_\mathtt{II}$ is still valid in $\mathcal{M}'$, by definition of the successor graph, and the last transition of this run is still $\mu_\mathtt{3}$. However, the number of produced \texttt{B}s is higher than the number of \texttt{B}s consumed. Such a run is not possible in $\mathcal{M}$, because $\mu_\mathtt{3}$ can be fired only when all \texttt{B}s have been consumed. Therefore $\pi_\mathtt{II} \not\in \mathfrak{L}(\mathcal{M})$ and from that $\mathfrak{L}(\mathcal{M}) \neq \mathfrak{L}(\mathcal{M}')$. \qed 
\end{proof}

\begin{theorem}\label{thr_or_vs_cr}
$\mathbb{OR} \incomparable \mathbb{CR}$
\end{theorem}

\begin{proof}
The property $\mathbb{CR} \not\sqsubseteq \mathbb{OR}$ follows directly from $\mathbb{OR} \incomparable \mathbb{CFR}$ (Figure~\ref{thr_or_vs_cfr}) and $\mathbb{CFR} \sqsubset \mathbb{CR}$ (Figure~\ref{thr_cfr_vs_cr}).

To show that $\mathbb{OR} \not\sqsubseteq \mathbb{CR}$, we find an $\mathcal{M} \in \mathbb{OR}$ such that $\forall \mathcal{M}' \in \mathbb{CR}$ holds that $\mathfrak{L}(\mathcal{M}) \neq \mathfrak{L}(\mathcal{M}')$. We use $\mathcal{M} \in \mathbb{OR}$ from Figure~\ref{or_example}. As it was stated in the proof of Figure~\ref{thr_mrs_vs_c}, both rules $\mu_\mathtt{1}, \mu_\mathtt{2} \in \mathcal{X}'$ of the $\mathcal{M}' \in \mathbb{CR}$. Due to the minimalistic nature of this system, it can be shown for any case of $\zeta'$ definition (assuming adding prohibited context with increased number of $\mathtt{A}$'s has no effect in this case), that $\mathfrak{L}(\mathcal{M}')$ contains more possible runs:

\vspace*{0.3cm}

\hspace*{1cm}\begin{tabularx}{1\textwidth}{m{5cm} m{3cm}}
$\zeta'(\mu_\mathtt{1}) = \emptyset \wedge \zeta'(\mu_\mathtt{2}) = \emptyset$& 
\begin{tikzpicture}
\node[state] (s0) {};
\node[state, right=1cm of s0] (s1) {$\mathtt{A}$};
\node[state, right=1cm of s1] (s2) {};
\node[state, right=1cm of s2] (s3) {$\mathtt{A}$};
\node[right=1cm of s3] (D) {$\ldots \not\in \mathfrak{L}(\mathcal{M})$};
\draw[shorten >=1mm,shorten <=1mm]
	    (s0) edge[above] node[label_node]{$\mu_\mathtt{1}$} (s1)
		(s1) edge[above] node[label_node]{$\mu_\mathtt{2}$} (s2)
		(s2) edge[above] node[label_node]{$\mu_\mathtt{1}$} (s3)
		(s3) edge[above] node[label_node]{$\mu_\mathtt{2}$} (D);
\end{tikzpicture} \\
$\zeta'(\mu_\mathtt{1}) = \mathtt{\{A\}} \wedge \zeta'(\mu_\mathtt{2}) = \emptyset$ & 
\begin{tikzpicture}
\node[state] (s0) {};
\node[state, right=1cm of s0] (s1) {$\mathtt{A}$};
\node[state, right=1cm of s1] (s2) {};
\node[state, right=1cm of s2] (s3) {$\mathtt{A}$};
\node[right=1cm of s3] (D) {$\ldots \not\in \mathfrak{L}(\mathcal{M})$};
\draw[shorten >=1mm,shorten <=1mm]
	    (s0) edge[above] node[label_node]{$\mu_\mathtt{1}$} (s1)
		(s1) edge[above] node[label_node]{$\mu_\mathtt{2}$} (s2)
		(s2) edge[above] node[label_node]{$\mu_\mathtt{1}$} (s3)
		(s3) edge[above] node[label_node]{$\mu_\mathtt{2}$} (D);
\end{tikzpicture} \\
$\zeta'(\mu_\mathtt{1}) = \emptyset \wedge \zeta'(\mu_\mathtt{2}) = \mathtt{\{A\}}$ & 
\begin{tikzpicture}
\node[state] (s0) {};
\node[state, right=1cm of s0] (s1) {$\mathtt{A}$};
\node[state, right=1cm of s1] (s2) {$\mathtt{A,A}$};
\node[state, right=1cm of s2] (s3) {$\mathtt{A,A,A}$};
\node[right=1cm of s3] (D) {$\ldots \not\in \mathfrak{L}(\mathcal{M})$};
\draw[shorten >=1mm,shorten <=1mm]
	    (s0) edge[above] node[label_node]{$\mu_\mathtt{1}$} (s1)
		(s1) edge[above] node[label_node]{$\mu_\mathtt{1}$} (s2)
		(s2) edge[above] node[label_node]{$\mu_\mathtt{1}$} (s3)
		(s3) edge[above] node[label_node]{$\mu_\mathtt{1}$} (D);
\end{tikzpicture} \\
$\zeta'(\mu_\mathtt{1}) = \mathtt{\{A\}} \wedge \zeta'(\mu_\mathtt{2}) = \mathtt{\{A\}}$ & 
\begin{tikzpicture}
\node[state] (s0) {};
\node[state, right=1cm of s0] (s1) {$\mathtt{A}$};
\node[state, right=1cm of s1] (s2) {$\mathtt{A}$};
\node[state, right=1cm of s2] (s3) {$\mathtt{A}$};
\node[right=1cm of s3] (D) {$\ldots \not\in \mathfrak{L}(\mathcal{M})$};
\draw[shorten >=1mm,shorten <=1mm]
	    (s0) edge[above] node[label_node]{$\mu_\mathtt{1}$} (s1)
		(s1) edge[above] node[label_node]{$\varepsilon$} (s2)
		(s2) edge[above] node[label_node]{$\varepsilon$} (s3)
		(s3) edge[above] node[label_node]{$\varepsilon$} (D);
\end{tikzpicture} \\
\end{tabularx}

\vspace*{0.3cm}

\noindent from that follows $\mathfrak{L}(\mathcal{M}) \neq \mathfrak{L}(\mathcal{M}')$. \qed 
\end{proof}

\begin{theorem}\label{thr_pr_vs_cr}
$\mathbb{PR} \incomparable \mathbb{CR}$
\end{theorem}

\begin{proof}
The property $\mathbb{PR} \not\sqsubseteq \mathbb{CR}$ follows directly from $\mathbb{OR} \incomparable \mathbb{CR}$ (Figure~\ref{thr_or_vs_cr}) and $\mathbb{OR} \sqsubset \mathbb{PR}$ (Figure~\ref{thr_or_vs_pr}). The property $\mathbb{CR} \not\sqsubseteq \mathbb{PR}$ follows directly from $\mathbb{CFR} \incomparable \mathbb{PR}$ (Figure~\ref{thr_cfr_vs_pr}) and $\mathbb{CFR} \sqsubset \mathbb{CR}$ (Figure~\ref{thr_cfr_vs_cr}). \qed 
\end{proof}

\begin{theorem}\label{thr_rr_vs_or}
$\mathbb{RR} \incomparable \mathbb{OR}$
\end{theorem}

\begin{proof}
To show that $\mathbb{RR} \not\sqsubseteq \mathbb{OR}$, we find an $\mathcal{M} \in \mathbb{RR}$ such that $\forall \mathcal{M}' \in \mathbb{OR}$ holds that $\mathfrak{L}(\mathcal{M}) \neq \mathfrak{L}(\mathcal{M}')$. We use the following $\mathcal{M} \in \mathbb{RR}$:

\begin{center}
$\mathcal{M} = 
		\left\{ 
			\begin{array}{l}
	  		\mathtt{M_0} = \emptyset, \hspace{0.2cm} \zeta = \left\{ 
          \begin{array}{l}
          \mu_\mathtt{1}.\mu_\mathtt{2}.\mu_\mathtt{1}.\mu_\mathtt{3}.\varepsilon^\omega
          \end{array} 
          \right\}, \\
	  		\mathcal{X} = \left\{ 
	  		 	\begin{array}{l}
	  		 	\mu_\mathtt{1}: \emptyset \to \mathtt{ \{A\} }, \mu_\mathtt{2}: \mathtt{ \{A\} } \to \emptyset, \mu_\mathtt{3}: \mathtt{ \{A\} } \to \mathtt{ \{B\} }
	  		 	\end{array} 
	  		  \right\} \\
	  		\end{array}
  		\right\}$
\end{center}

\noindent with $\mathcal{S} = \mathtt{\{ A,B \}}$. In this system, rules can be used only in a particular finite sequence, followed by infinite application of empty rule $\varepsilon$. Therefore, the set of runs $\mathfrak{L}(\mathcal{M})$ contains only one possible run $\pi$:

{\centering
\begin{tikzpicture}

\node[state] (s0) {};
\node[state, right=1cm of s0] (s1) {$\mathtt{A}$};
\node[state, right=1cm of s1] (s2) {};
\node[state, right=1cm of s2] (s3) {$\mathtt{A}$};
\node[state, right=1cm of s3] (s4) {$\mathtt{B}$};
\node[right=1cm of s4] (D) {$\ldots$};

\draw[shorten >=1mm,shorten <=1mm]
	    (s0) edge[above] node[label_node]{$\mu_\mathtt{1}$} (s1)
		(s1) edge[above] node[label_node]{$\mu_\mathtt{2}$} (s2)
		(s2) edge[above] node[label_node]{$\mu_\mathtt{1}$} (s3)
		(s3) edge[above] node[label_node]{$\mu_\mathtt{3}$} (s4)
		(s4) edge[above] node[label_node]{$\varepsilon$} (D);
\end{tikzpicture}
\\}

\noindent Since the run $\pi$ enforces that rule $\mu_\mathtt{1} \in \mathcal{X}'$, it is inevitable that, for example, run:

{\centering
\begin{tikzpicture}

\node[state] (s0) {};
\node[state, right=1cm of s0] (s1) {$\mathtt{A}$};
\node[state, right=1cm of s1] (s2) {$\mathtt{A,A}$};
\node[state, right=1cm of s2] (s3) {$\mathtt{A,A,A}$};
\node[state, right=1cm of s3] (s4) {$\mathtt{A,A,A,A}$};
\node[right=1cm of s4] (D) {$\ldots$};

\draw[shorten >=1mm,shorten <=1mm]
	    (s0) edge[above] node[label_node]{$\mu_\mathtt{1}$} (s1)
		(s1) edge[above] node[label_node]{$\mu_\mathtt{1}$} (s2)
		(s2) edge[above] node[label_node]{$\mu_\mathtt{1}$} (s3)
		(s3) edge[above] node[label_node]{$\mu_\mathtt{1}$} (s4)
		(s4) edge[above] node[label_node]{$\mu_\mathtt{1}$} (D);
\end{tikzpicture}
\\}

\noindent belongs to $\mathfrak{L}(\mathcal{M'})$ because it can be applied in the first step and its subsequent applications cannot be limited in the ordered system (follows from irreflexivity of the order $\zeta'$).

To show that $\mathbb{OR} \not\sqsubseteq \mathbb{RR}$, we find an $\mathcal{M} \in \mathbb{OR}$ such that $\forall \mathcal{M}' \in \mathbb{RR}$ holds that $\mathfrak{L}(\mathcal{M}) \neq \mathfrak{L}(\mathcal{M}')$. We use $\mathcal{M} \in \mathbb{OR}$ from Figure~\ref{or_example}. Let $\mathcal{A} = \{ \mu_\mathtt{1}^\mathtt{n}.\mu_\mathtt{2}^\mathtt{n}.\varepsilon^\omega ~|~ \mathtt{n} \in \mathbb{N} \}$ be an $\omega$-language over rules $\mu_\mathtt{1},\mu_\mathtt{2},\varepsilon$. For the system $\mathcal{M}$, it holds that $\forall \pi \in \mathfrak{L}(\mathcal{M}). \overrightarrow{\pi} \in \mathcal{A}$. However, language $\mathcal{A}$ is obviously not $\omega$-regular, therefore such an $\mathcal{M}' \in \mathbb{RR}$ does not exist. \qed 
\end{proof}

\begin{theorem}\label{thr_rr_vs_pr}
$\mathbb{RR} \incomparable \mathbb{PR}$
\end{theorem}

\begin{proof}
To show that $\mathbb{RR} \not\sqsubseteq \mathbb{PR}$, we find an $\mathcal{M} \in \mathbb{RR}$ such that $\forall \mathcal{M}' \in \mathbb{PR}$ holds that $\mathfrak{L}(\mathcal{M}) \neq \mathfrak{L}(\mathcal{M}')$. We use the $\mathcal{M} \in \mathbb{RR}$ from the proof of Figure~\ref{thr_rr_vs_or}. The run $\pi$ actually enforces all the rules $\mu_\mathtt{1},\mu_\mathtt{2},\mu_\mathtt{3}$ to be present in $\mathcal{X}'$ due to steps one, two, and four, respectively. Let us focus on the successor function $\zeta'$ for rule $\mu_\mathtt{1}$ w.r.t. rules $\mu_\mathtt{2},\mu_\mathtt{3}$. There are four options:

\vspace*{0.3cm}

\hspace*{1.5cm}\begin{tabularx}{1\textwidth}{m{4cm} m{10cm}}
$\mu_\mathtt{2} \not\in \zeta'(\mu_\mathtt{1}) \wedge \mu_\mathtt{3} \not\in \zeta'(\mu_\mathtt{1})$ & 
$\ldots$ run $\pi \not\in \mathfrak{L}(\mathcal{M}')$ because $\mu_\mathtt{2}$ cannot be used after $\mu_\mathtt{1}$\\
$\mu_\mathtt{2} \not\in \zeta'(\mu_\mathtt{1}) \wedge \mu_\mathtt{3} \in \zeta'(\mu_\mathtt{1})$ & 
$\ldots$ run $\pi \not\in \mathfrak{L}(\mathcal{M}')$ because $\mu_\mathtt{2}$ cannot be used after $\mu_\mathtt{1}$\\
$\mu_\mathtt{2} \in \zeta'(\mu_\mathtt{1}) \wedge \mu_\mathtt{3} \not\in \zeta'(\mu_\mathtt{1})$ & 
$\ldots$ run $\pi \not\in \mathfrak{L}(\mathcal{M}')$ because $\mu_\mathtt{3}$ cannot be used after $\mu_\mathtt{1}$\\
$\mu_\mathtt{2} \in \zeta'(\mu_\mathtt{1}) \wedge \mu_\mathtt{3} \in \zeta'(\mu_\mathtt{1})$ & 
\begin{tikzpicture}[scale=0.7]
\node[state] (s0) {};
\node[state, right=1cm of s0] (s1) {$\mathtt{A}$};
\node[state, right=1cm of s1] (s2) {$\mathtt{B}$};
\node[state, right=1cm of s2] (s3) {$\mathtt{A,B}$};
\node[state, right=1cm of s3] (s4) {$\mathtt{B}$};
\node[right=1cm of s4] (D) {$\ldots \pi'$};
\draw[shorten >=1mm,shorten <=1mm]
	    (s0) edge[above] node[label_node]{$\mu_\mathtt{1}$} (s1)
		(s1) edge[above] node[label_node]{$\mu_\mathtt{3}$} (s2)
		(s2) edge[above] node[label_node]{$\mu_\mathtt{1}$} (s3)
		(s3) edge[above] node[label_node]{$\mu_\mathtt{2}$} (s4)
		(s4) edge[above] node[label_node]{} (D);
\end{tikzpicture} \\
\end{tabularx}

\vspace*{0.3cm}

\noindent with $\pi' \in \mathfrak{L}(\mathcal{M'})$ but $\pi' \not\in \mathfrak{L}(\mathcal{M})$. In each case holds $\mathfrak{L}(\mathcal{M}) \neq \mathfrak{L}(\mathcal{M}')$.

The property $\mathbb{PR} \not\sqsubseteq \mathbb{RR}$ follows directly from $\mathbb{RR} \incomparable \mathbb{OR}$ (Figure~\ref{thr_rr_vs_or}) and $\mathbb{OR} \sqsubset \mathbb{PR}$ (Figure~\ref{thr_or_vs_pr}). \qed 
\end{proof}

\newpage
\begin{theorem}\label{thr_rr_vs_cfr}
$\mathbb{RR} \incomparable \mathbb{CFR}$
\end{theorem}

\begin{proof}
To show that $\mathbb{RR} \not\sqsubseteq \mathbb{CFR}$, we find an $\mathcal{M} \in \mathbb{RR}$ such that $\forall \mathcal{M}' \in \mathbb{CFR}$ holds that $\mathfrak{L}(\mathcal{M}) \neq \mathfrak{L}(\mathcal{M}')$. We use the $\mathcal{M} \in \mathbb{RR}$ from the proof of Figure~\ref{thr_rr_vs_or}. As stated in the proof of Figure~\ref{thr_rr_vs_pr}, all rules $\mu_\mathtt{1},\mu_\mathtt{2},\mu_\mathtt{3}$ belong to $\mathcal{X}'$. Rules $\mu_\mathtt{2},\mu_\mathtt{3}$ are concurrent and there are three options how to resolve the concurrency:

\vspace*{0.3cm}

\hspace{1.5cm}\begin{tabularx}{1\textwidth}{m{4cm} m{3cm}}
\multicolumn{1}{c}{\makecell{concurrency is \\ not resolved}} & 
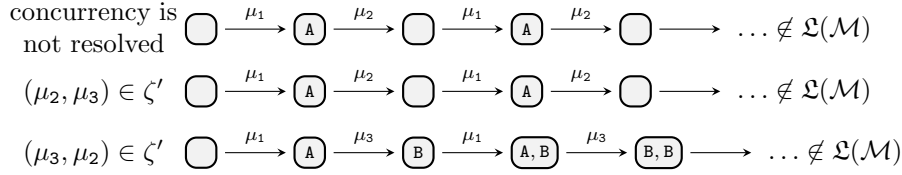
\begin{tikzpicture}[scale=0.8]
\node[state] (s0) {};
\node[state, right=1cm of s0] (s1) {$\mathtt{A}$};
\node[state, right=1cm of s1] (s2) {};
\node[state, right=1cm of s2] (s3) {$\mathtt{A}$};
\node[state, right=1cm of s3] (s4) {};
\node[right=1cm of s4] (D) {$\ldots \not\in \mathfrak{L}(\mathcal{M})$};
\draw[shorten >=1mm,shorten <=1mm]
	    (s0) edge[above] node[label_node]{$\mu_\mathtt{1}$} (s1)
		(s1) edge[above] node[label_node]{$\mu_\mathtt{2}$} (s2)
		(s2) edge[above] node[label_node]{$\mu_\mathtt{1}$} (s3)
		(s3) edge[above] node[label_node]{$\mu_\mathtt{2}$} (s4)
		(s4) edge[above] node[label_node]{} (D);
\end{tikzpicture} \\
\multicolumn{1}{c}{$(\mu_\mathtt{2}, \mu_\mathtt{3}) \in \zeta'$} & 
\begin{tikzpicture}[scale=0.8]
\node[state] (s0) {};
\node[state, right=1cm of s0] (s1) {$\mathtt{A}$};
\node[state, right=1cm of s1] (s2) {};
\node[state, right=1cm of s2] (s3) {$\mathtt{A}$};
\node[state, right=1cm of s3] (s4) {};
\node[right=1cm of s4] (D) {$\ldots \not\in \mathfrak{L}(\mathcal{M})$};
\draw[shorten >=1mm,shorten <=1mm]
	    (s0) edge[above] node[label_node]{$\mu_\mathtt{1}$} (s1)
		(s1) edge[above] node[label_node]{$\mu_\mathtt{2}$} (s2)
		(s2) edge[above] node[label_node]{$\mu_\mathtt{1}$} (s3)
		(s3) edge[above] node[label_node]{$\mu_\mathtt{2}$} (s4)
		(s4) edge[above] node[label_node]{} (D);
\end{tikzpicture} \\
\multicolumn{1}{c}{$(\mu_\mathtt{3}, \mu_\mathtt{2}) \in \zeta'$} & 
\begin{tikzpicture}[scale=0.8]
\node[state] (s0) {};
\node[state, right=1cm of s0] (s1) {$\mathtt{A}$};
\node[state, right=1cm of s1] (s2) {$\mathtt{B}$};
\node[state, right=1cm of s2] (s3) {$\mathtt{A,B}$};
\node[state, right=1cm of s3] (s4) {$\mathtt{B,B}$};
\node[right=1cm of s4] (D) {$\ldots \not\in \mathfrak{L}(\mathcal{M})$};
\draw[shorten >=1mm,shorten <=1mm]
	    (s0) edge[above] node[label_node]{$\mu_\mathtt{1}$} (s1)
		(s1) edge[above] node[label_node]{$\mu_\mathtt{3}$} (s2)
		(s2) edge[above] node[label_node]{$\mu_\mathtt{1}$} (s3)
		(s3) edge[above] node[label_node]{$\mu_\mathtt{3}$} (s4)
		(s4) edge[above] node[label_node]{} (D);
\end{tikzpicture} \\
\end{tabularx}

\vspace*{0.3cm}

\noindent In each case holds $\mathfrak{L}(\mathcal{M}) \neq \mathfrak{L}(\mathcal{M}')$.

To show that $\mathbb{CFR} \not\sqsubseteq \mathbb{RR}$, we find an $\mathcal{M} \in \mathbb{CFR}$ such that $\forall \mathcal{M}' \in \mathbb{RR}$ holds that $\mathfrak{L}(\mathcal{M}) \neq \mathfrak{L}(\mathcal{M}')$. We use the following $\mathcal{M} \in \mathbb{CFR}$:

\begin{center}
$\mathcal{M} = 
		\left\{ 
			\begin{array}{l}
	  		\mathtt{M_0} = \mathtt{ \{C\} },\hspace{0.2cm} \zeta = \left\{ 
          \begin{array}{l}
          (\mu_\mathtt{4}, \mu_\mathtt{3})
          \end{array} 
          \right\}, \\
	  		\mathcal{X} = \left\{ 
	  		 	\begin{array}{l}
	  		 	\mu_\mathtt{1}: \mathtt{ \{C\} } \to \mathtt{ \{A,C\} }, \mu_\mathtt{2}: \mathtt{ \{C\} } \to \mathtt{ \{B\} }\\
	  		 	\mu_\mathtt{3}: \mathtt{ \{A,B\} } \to \mathtt{ \{B\} }, \mu_\mathtt{4}: \mathtt{ \{B\} } \to \emptyset
	  		 	\end{array} 
	  		  \right\} \\
	  		\end{array}
  		\right\}$
\end{center}

\noindent with $\mathcal{S} = \mathtt{\{ A,B,C \}}$. In this system, a number $\mathtt{n} \in \mathbb{N}$ of elements \texttt{A} is generated, followed by the change of \texttt{C} to a \texttt{B} and finally, all elements \texttt{A} are discarded, including symbol \texttt{B}. This is ensured by giving priority to rule $\mu_\mathtt{3}$ over rule $\mu_\mathtt{4}$. Please note the concurrency between rules $\mu_\mathtt{1}$ and $\mu_\mathtt{2}$ is not resolved (which is allowed by definition). Since runs:

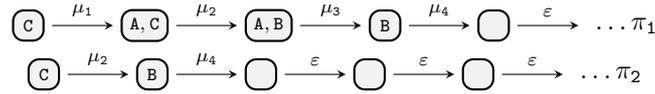
{\centering
\begin{tikzpicture}

\node[state] (s0) {$\mathtt{C}$};
\node[state, right=1cm of s0] (s1) {$\mathtt{A,C}$};
\node[state, right=1cm of s1] (s2) {$\mathtt{A,B}$};
\node[state, right=1cm of s2] (s3) {$\mathtt{B}$};
\node[state, right=1cm of s3] (s4) {};
\node[right=1cm of s4] (D) {$\ldots \pi_\mathtt{1}$};

\draw[shorten >=1mm,shorten <=1mm]
	    (s0) edge[above] node[label_node]{$\mu_\mathtt{1}$} (s1)
		(s1) edge[above] node[label_node]{$\mu_\mathtt{2}$} (s2)
		(s2) edge[above] node[label_node]{$\mu_\mathtt{3}$} (s3)
		(s3) edge[above] node[label_node]{$\mu_\mathtt{4}$} (s4)
		(s4) edge[above] node[label_node]{$\varepsilon$} (D);
\end{tikzpicture}

\begin{tikzpicture}

\node[state] (s0) {$\mathtt{C}$};
\node[state, right=1cm of s0] (s1) {$\mathtt{B}$};
\node[state, right=1cm of s1] (s2) {};
\node[state, right=1cm of s2] (s3) {};
\node[state, right=1cm of s3] (s4) {};
\node[right=1cm of s4] (D) {$\ldots \pi_\mathtt{2}$};

\draw[shorten >=1mm,shorten <=1mm]
	    (s0) edge[above] node[label_node]{$\mu_\mathtt{2}$} (s1)
		(s1) edge[above] node[label_node]{$\mu_\mathtt{4}$} (s2)
		(s2) edge[above] node[label_node]{$\varepsilon$} (s3)
		(s3) edge[above] node[label_node]{$\varepsilon$} (s4)
		(s4) edge[above] node[label_node]{$\varepsilon$} (D);
\end{tikzpicture}
\\}

\noindent are both allowed, both $\mu_\mathtt{2}, \mu_\mathtt{4} \in \mathcal{X}'$ and also $\mu_\mathtt{I}, \mu_\mathtt{III} \in \mathcal{X}'$ such that $\mathtt{A} \in \mu_\mathtt{I}^\bullet$ (from the first step) and $\mathtt{A} \in~ ^\bullet\mu_\mathtt{III}$ (from the third step).

Let $\mathcal{A} = \{ \mu_\mathtt{I}^\mathtt{n}.\mu_\mathtt{2}.\mu_\mathtt{III}^\mathtt{n}.\mu_\mathtt{4}.\varepsilon^\omega ~|~ \mathtt{n} \in \mathbb{N} \}$ be an $\omega$-language over rules $\mu_\mathtt{I},\mu_\mathtt{2},\mu_\mathtt{III},\mu_\mathtt{4},\varepsilon$. For the system $\mathcal{M}$, it holds that $\forall \pi \in \mathfrak{L}(\mathcal{M}). \overrightarrow{\pi} \in \mathcal{A}$ (such that instead of rules $\mu_\mathtt{I}, \mu_\mathtt{III}$ their possible instances $\mu_\mathtt{1}, \mu_\mathtt{3}$ are used, respectively). However, language $\mathcal{A}$ is obviously not $\omega$-regular, therefore such an $\mathcal{M}' \in \mathbb{RR}$ does not exist. \qed 
\end{proof}

\begin{theorem}\label{thr_rr_vs_cr}
$\mathbb{RR} \incomparable \mathbb{CR}$
\end{theorem}

\begin{proof}
To show that $\mathbb{RR} \not\sqsubseteq \mathbb{CR}$, we find an $\mathcal{M} \in \mathbb{RR}$ such that $\forall \mathcal{M}' \in \mathbb{CR}$ holds that $\mathfrak{L}(\mathcal{M}) \neq \mathfrak{L}(\mathcal{M}')$. We use the $\mathcal{M} \in \mathbb{RR}$ from the proof of Figure~\ref{thr_rr_vs_or}. From the only possible run $\pi$ it follows that rules $\mu_\mathtt{2}, \mu_\mathtt{3}$ need to be applicable in the context of element \texttt{A} while it is only possible candidate to be in their prohibited context. Since that cannot be done, there will be some additional runs in $\mathfrak{L}(\mathcal{M}')$. From that $\mathfrak{L}(\mathcal{M}) \neq \mathfrak{L}(\mathcal{M}')$.

The property $\mathbb{CR} \not\sqsubseteq \mathbb{RR}$ follows directly from $\mathbb{RR} \incomparable \mathbb{CFR}$ (Figure~\ref{thr_rr_vs_cfr}) and $\mathbb{CFR} \sqsubset \mathbb{CR}$ (Figure~\ref{thr_cfr_vs_cr}). \qed 
\end{proof}

\subsection{Expressive power}

In this section, we discuss \emph{expressive power} for some of the defined classes. While the generative power discussed above states about the particular sets of runs that can be generated, the expressive power studies \emph{functions} which can be computed by the formalism.

A multiset rewriting system (with or without regulation) \emph{strongly computes} a numerical function $f: \mathbb{N} \to \mathbb{N}$ by starting in an initial multiset $\mathtt{M_0}$ containing an input element $\mathtt{I}$ with $\mathtt{M_0}(\mathtt{I}) = n$ and always reaches a final multiset $\mathtt{M_f}$ (all subsequent multisets are equal) containing an output element $\mathtt{O}$ with $\mathtt{M_f}(\mathtt{O}) = f(n)$. In order for the system to be correct, there cannot exist a run where for the final multiset $\mathtt{M_f}$ holds that $\mathtt{M_f}(\mathtt{O}) \neq f(n)$.

Since this definition does not accommodate nondeterminism nicely (which is natural for multiset rewriting), we say a multiset rewriting system \emph{weakly computes} a numerical function $f: \mathbb{N} \to \mathbb{N}$ if there exists a run with $\mathtt{M_f}(\mathtt{O}) = f(n)$ and for any other run holds that $\mathtt{M_f}(\mathtt{O}) \leq f(n)$. This definition is an adoption of weak computability for Petri nets~\cite{leroux2014functions}.

\newpage
\begin{theorem}
$\mathbb{MRS}$ weakly compute all computable numerical functions.\label{result_ii}
\end{theorem}

For the proof, we refer to~\cite{cervesato1994petri,cervesato1995petri} where it is shown that expressive power of $\mathbb{MRS}$ is equivalent to expressive power of Petri Nets, because they are both based on simple operations on positive integer counters -- decrements and increments. More precisely, they cannot compute all recursive functions because they lack zero-tests which makes them less expressive than Turing machine  -- in particular, they weakly compute numerical functions~\cite{leroux2014functions}.

\begin{theorem}
$\mathbb{CR}$ and $\mathbb{CFR}$ compute all computable numerical functions. \label{result_iii}
\end{theorem}

\begin{proof}
Assuming the Church-Turing thesis is correct, the simplest way how to prove this theorem is in terms of Minsky register machine~\cite{minsky1967computation}. Minsky machine has a number of registers storing arbitrary large numbers. A program is a sequence of instructions manipulating the registers. It was shown that two register machines with the following instruction set can compute all computable numerical functions (shown by equivalence to Turing machine~\cite{shepherdson1963computability}):

\begin{tabular}{l l}
  $l_1$ & $: \mathtt{Com}_1$;\\
  $l_2$ & $: \mathtt{Com}_2$;\\
   & \vdots \\
  $l_m$ & $: \mathtt{Com}_m$;\\
  $l_{m+1}$ & $: \mathtt{halt}$;\\
\end{tabular}

\noindent where $m \leq 0$ and every $\mathtt{Com}_i$ is one of two types:

\begin{itemize}
  \item Type I: $c_j := c_j + 1; ~\mathtt{goto}~ l_p$
  \item Type II: $\mathtt{if}~ c_j = 0 ~\mathtt{then~goto}~ l_p ~\mathtt{else}~ c_j := c_j -1; ~\mathtt{goto}~ l_q$
\end{itemize}

\noindent with counter index $j \in \{1,2\}$ and $p,q \in \{1, \ldots, m+1\}$.

Computation of the machine with the counters set to initial values (with $c_1 = n$) starts at $l_1$ and proceeds consistently with the intuitive semantics of the instructions. We say that the machine \emph{halts} if the computation eventually reaches the \texttt{halt} instruction. Finally, the result of computation is the value of counter $c_2$.

Thus, if a register machine can be converted into an equivalent conditional (resp. concurrent-free) regulated MRS, we see that $\mathbb{CR}$ (resp. $\mathbb{CFR}$) can compute all computable numerical functions. First, we show this for $\mathbb{CR}$. The proof for $\mathbb{CFR}$ is identical except for a minor modification specific to its regulation approach.

To represent a register machine as an $\mathcal{M} \in \mathbb{CR}$, we represent the two registers by two elements $c_1, c_2$ (such that intially $c_1$ has $n$ repretitions and $c_2$ represents the output counter). We use elements $l_1, \ldots, l_m, l_{m+1}$ to represent the position of the program. The instructions are represented as follows:

\begin{itemize}
  \item Type I: introduce rule $\mu_\mathtt{i}: \{ l_i \} \rightarrow \{ l_p, c_j \}$
  \item Type II: introduce rules $\mu_\mathtt{i}: \{ l_i \} \rightarrow \{ l_p \}$ and $\overline{\mu}_\mathtt{i}: \{ l_i, c_j \} \rightarrow \{ l_q \}$ and set $\zeta(\mu_\mathtt{i}) = \{ c_j \}$
\end{itemize}

Defining prohibited context to the rule $\mu_\mathtt{i}$ in Type II instruction effectively prioritises rule $\overline{\mu}_\mathtt{i}$ until the counter $c_j$ is not zero. Finally, the instruction \texttt{halt} is represented by the application of rule $\varepsilon$.

This shows that a register machine can be converted into $\mathbb{CR}$. In the case of $\mathbb{CFR}$, the only modification is in the definition of regulation $\zeta$ in the Type II instruction -- we require that $(\mu_\mathtt{i}, \overline{\mu}_\mathtt{i}) \in \zeta$, which ensures that $\overline{\mu}_\mathtt{i}$ has priority until the counter $c_j$ is not zero. \qed 
\end{proof}

\section{Conclusions}

In this paper, we introduced regulated multiset rewriting systems and discussed relationships among them, showing the usefulness of individual classes as most of them are incomparable. The considered regulation mechanisms were mostly inspired by the needs encountered in modelling of biological phenomenons. As our future work, we would like to explore other mechanisms, for example those based on other subclasses of $\omega$ languages. We also want to investigate the expressive power of regulated MRS in more detail and focus on the complexity of common problems in system biology, such as reachability. Finally, we will apply multiset rewriting regulations to a rule-based Biochemical Space language~\cite{trojak2020plosone}, which can be seen as an instantiation of MRS, and present biological case studies of the regulations using this language.

\bibliographystyle{plainurl}
\bibliography{papers}

\end{document}